\documentclass[11pt]{article}
\usepackage[margin=1in]{geometry}
\usepackage{graphicx} 
\usepackage{amsmath, amsfonts,amssymb}
\usepackage{amsthm}
\usepackage{amstext}
\usepackage{amsopn}
\usepackage{appendix}
\usepackage{subfigure}
\usepackage{tikz}
\usepackage{booktabs,multirow}

\usepackage{multirow}
\usepackage{array}
\usepackage{hyphenat}  
\usepackage[shortlabels]{enumitem}

\usepackage{graphicx}
\usepackage{bm}
\newtheorem{theorem}{Theorem}[section]
\newtheorem{lemma}[theorem]{Lemma}
\newtheorem{proposition}[theorem]{Proposition}

\newtheorem{remark}[theorem]{Remark}
\newcommand{\circled}[1]{\raisebox{.5pt}{\textcircled{\raisebox{-.9pt}{#1}}}}

\def\MM{\mathcal{M}}
\def\TT{\mathcal{T}}
\def\DD{\mathcal{D}}
\theoremstyle{definition}
\newtheorem{definition}[theorem]{Definition}

\numberwithin{equation}{section}
\title{{\vspace{-80pt}} Impact of network connectivity on the dynamics of populations in stream environments
\begin{center}
    {\small In memory of Professor Pierre Magal}
\end{center}
}
\date{}
\author{~}

\begin{document}
\author{Tung D. Nguyen\footnote{Department of Mathematics, University of California Los Angeles, Los Angeles, CA 90024, USA}, Tingting Tang\footnote{Department of Mathematics and Statistics, San Diego State University, San Diego, CA 92182, USA}, Amy Veprauskas\footnote{Department of Mathematics, University of Louisiana at Lafayette, Lafayette, LA 70501, USA}, Yixiang Wu\footnote{Department of Mathematical Sciences, Middle Tennessee State University, Murfreesboro, Tennessee 37132, USA}, and Ying Zhou\footnote{Department of Mathematical Sciences, Lafayette College, Easton, PA 18042, USA}}

\maketitle

\begin{abstract}
    We consider the impact of network connectivity on the dynamics of a population in a stream environment. The population is modeled using a graph theoretical framework, with habitats represented by isolated patches. We introduce a change in connectivity into the model through the addition of a bi-directional or one-directional edge between two patches and examine the impact of this edge modification on the metapopulation growth rate and the network biomass. Our main results indicate that adding a bi-directional edge often decreases both measures, while the effect of adding an one-directional edge is more intricate and dependent on the model parameters. We establish complete analytical results for stream networks of three patches, and provide some generalizations and conjectures for more general stream networks of $n$ patches. These conjectures are supported with numerical simulations.
    
\end{abstract}


\section{Introduction}

    



Lotic environments, such as streams and rivers, play a vital role in the ecosystem by connecting with adjacent terrestrial and marine environments,  transporting nutrients and sediments, and providing essential habitats for an array of diverse organisms \cite{giller1998biology}. Streams and rivers are unique habitats with  unidirectional flows and dendritic structures that shape the dwelling and dispersal of organisms. They are also highly heterogeneous in space and time, affected by seasonal changes, natural disturbances, and human activities \cite{pelletier1996dynamics,tonkin2018role}. Understanding the complex mechanisms that impact population persistence and abundance in stream networks is both challenging and valuable. 


Various mathematical models have been employed to study population dynamics in stream environments. 
One common approach is to apply a graph-theoretical framework to represent the stream network \cite{chen2022invasion, chen2023impact, chen2024evolution,chen2023evolution, Ducrot2022,ge2023global, hamida2017evolution,jiang2020two, jiang2021three, liu2025global, liu2024dynamics, lou2019ideal, nguyen2023impact, nguyen2023maximizing, noble2015evolution}. In these models, stream habitats are represented as isolated nodes (or patches), while directed edges between the nodes indicate the direction of organism movement.  The weight of each edge corresponds to the dispersal rate of individuals, capturing the flow of organisms across the network. 
To investigate the population dynamics in stream networks, an ordinary differential equation (ODE) can be applied to each patch, describing the temporal variation in the local population. Meanwhile, the migration of organisms between patches is typically modeled using a matrix, which tracks the movement of individuals across the network. Alternative approaches to modeling stream networks include reaction-diffusion (PDE) models \cite{du2020fisher, huang2016r0,jin2011seasonal,jin2019population,lam2016emergence,lou2014evolution,lou2015evolution,lutscher2006effects,lutscher2005effect,vasilyeva2024evolution}. 

In recent works \cite{jiang2020two, jiang2021three}, the authors proposed three different configurations of stream networks of three nodes and studied how the flow drift rate and  random dispersal rate affect the competitive exclusion and coexistence of two stream species. The dynamics of these three-patch models over a broader range of parameters are elucidated in \cite{chen2023impact, liu2024dynamics}. When there are $n-$patches configured along a straight line, the authors in \cite{chen2022invasion,chen2024evolution,chen2023evolution,liu2025global} have considered the influence of dispersal rates, downstream boundary conditions, and heterogeneity in drift rate and carrying capacity on the competition outcome of two stream species. The above mentioned results greatly generalized previous studies \cite{ge2023global,hamida2017evolution,lou2019ideal,noble2015evolution} on stream patch models with two nodes. 
In \cite{nguyen2023maximizing}, the authors study the optimal distribution of resources in  stream networks and the results indicate that to maximize the metapopulation growth rate, the sources should be concentrated in the most
downstream locations. Conversely, to maximize the total biomass, the sources should be concentrated on the most upstream locations. In \cite{nguyen2023impact}, the authors further showed that the stream species whose distribution of resources  maximizes the total biomass may have competitive advantage. 


Studies have highlighted effects of connectivity and dispersal in stream networks on metapopulation dynamics and spatial distributions of populations \cite{fagan2002connectivity, gonzalez2019effects}. Network structure governs the dispersal of aquatic species in stream networks, and significantly impacts local population dynamics, genetic diversity, and community compositions \cite{moore2015bidirectional,thompson2017loss,tonkin2018role,tonkin2018metacommunities,wofford2005influence}.  Network connectivity is dynamic and can be affected by seasonal fluctuations, natural disturbances, and landscape changes \cite{perry2019does, zeigler2014transient}. Due to the hierarchical and dendritic structure of stream networks \cite{altermatt2013diversity}, changes in network connectivity can be coupled with source-sink dynamics and network structure to jointly impact population persistence and abundance. 
Numerical simulations have shown that population persistence  is negatively affected by network connectivity or connectivity dendritic networks can promote local extinction and diminish metapopulation size \cite{labonne2008linking}. 
In contrast, it has also been found that increasing the patch degree (the number of connections per patch) reduces the extinction probability and benefits organism abundance \cite{arancibia2022network}. Some other studies argue that movement obstacles can either increase or decrease population growth rate \cite{samia2015connectivity}.


In the current paper, the main question we address is how the connectivity of a stream network impacts the population persistence of a single species. More specifically, we examine the effect of {\em edge modification} on two persistence measures: the metapopulation growth rate and the network biomass (of the stable positive equilibrium). We consider two types of edge modification: bi-directional edge modification where a pair of reversible edges is added between two nodes and one-directional edge modification where only a single directed edge is added between two nodes. We primarily focus on two configurations of the 3-patch stream network in which the edge modification makes sense: the tributary configuration where there is two upstream nodes and one downstream node, and the distributary configuration where there is one upstream nodes and two downstream nodes (Figures \ref{fig-tributary} and \ref{fig-distributary}). Moreover, we extend some of these results to a general stream network of $n$ nodes, as was defined in \cite{nguyen2023maximizing}.

Our analytical and simulation results indicate that increasing network connectivity by adding one or more edges has mixed effects on the population persistence. Below, we list the main biological insights from our findings:  
\begin{enumerate}
    \item [(i)] {\em Adding a pair of bi-directional edges often decreases both persistence measures  (see Theorems \ref{thm:rho'-3patch}, \ref{thm:biomass_2dir_distributary}, and \ref{thm:rho'-npatch})}. There is an exception to this observation with the tributary configuration, where under some specific conditions, the biomass may increase.
    \item [(ii)] {\em Adding a one-directional edge from a faster growing node (higher intrinsic growth rate) to a slower growing node (lower intrinsic growth rate) tends to decrease the metapopulation growth rate and vice versa.} We provide the proofs for the 3-patch stream networks (see Theorems \ref{thm:rho'_onedir_tributary} and \ref{thm:rho'_onedir_distributary}) and some simulation results for stream networks with more than 3 patches.
    \item [(iii)] {\em Conversely, adding a one-directional edge from a faster growing node  to a slower growing node tends to increase the network biomass when the rate on the added edge is sufficiently small and vice versa.} Again, we provide the proofs for the 3-patch stream networks (see Theorems \ref{thm:biomass_1dir_tributary} and \ref{thm:biomass_1dir_distributary}) and some simulation results for larger stream networks. 
\end{enumerate}
Notably, the last two insights suggest that there seems to be a trade-off between prioritizing the metapopulation growth rate and the network biomass when the rate of the added edge is small. This trade-off was also observed in our previous study where we examined how resource distributions impact these two persistence measures \cite{nguyen2023maximizing}.

The paper is organized as follows. In Section \ref{sec:model}, we formulate our ODE model for a single species on stream networks, and define the two types of edge modifications on three-patch stream networks. In Sections \ref{sec:growth} and \ref{sec:biomass}, we investigate the effects of edge modifications on the metapopulation growth rate and network biomass respectively. In Section \ref{sec:homflow}, we extend some of our results to stream networks with an arbitrary number of patches. Finally, in Section \ref{sec:discussion}, we summarize our findings and make conjectures about more general stream networks, which are backed up with simulation results.

\section{Model formulation}\label{sec:model}
Let $n$ be a positive integer representing the fixed number of patches (or nodes) and $u_i=u_i(t), 1\le i\le n$, denote the population scale (size or density) of a certain species of study in patch $i$ at time $t\ge 0$. The model we consider is of the form
 \begin{equation}\label{eq-system}
u_i' = r_iu_i\Big(1-\frac{u_i}{K_i}\Big) + \sum_{j=1}^n \big(\ell_{ij} u_j - \ell_{ji}u_i\big), \quad i=1,\ldots,n,
\end{equation}
where population growth in each patch follows the logistic equation with $r_i$ being the local population growth rate and $K_i$ being the local carrying capacity at patch $i$. Here $\ell_{ij}\ge 0$ denotes the movement rate of the individuals from patch $j$ to patch $i$, for $1\le i,j\le n$ and $i\not= j$. 
We always assume $\ell_{ii}=0$ for all $i$. 
The first term in the sum above, $\sum_{j} \ell_{ij} u_j$, tracks all incoming movements (flux in) to patch $i$ while the second term, $\sum_j \ell_{ji}u_i$, sums all outgoing movements (flux out) departing from patch $i$. 

All movement coefficients in \eqref{eq-system} can be associated with a \textit{movement network} $G$. We denote such a movement network as $(G, L)$, where the $n\times n$ \textit{connection matrix} $L$, whose off-diagonal entries are $\ell_{ij}$ and diagonal entries are $-\sum_j \ell_{ji}$, is as follows:

\begin{equation}\label{eq-L}
L:=\begin{pmatrix}
-\sum_j \ell_{j1} & \ell_{12} & \cdots & \ell_{1n}\\
\ell_{21} & -\sum_j \ell_{j2} & \cdots & \ell_{2n}\\
\vdots & \vdots & \ddots & \vdots\\
\ell_{n1} & \ell_{n2} & \cdots & -\sum_j \ell_{jn}
\end{pmatrix}.
\end{equation}
 We assume that the movement network $G$ is strongly connected, i.e.,  $L$ is irreducible.

We focus on two measures of population persistence, namely the metapopulation growth rate and the total biomass. The \textit{metapopulation growth rate} $\rho$ of \eqref{eq-system} determines the metapopulation growth rate when the population is small. Let $s(A)$ denote the \textit{spectral bound} of a matrix $A$, i.e. 
\[
s(A):=\max\Big\{\mathrm{Re} \lambda : \; \lambda \ \mbox{is an eigenvalue of }{A} \Big\}.
\]
Then the metapopulation growth rate is defined as the spectral bound of the Jacobian matrix $J$ for the linearization of \eqref{eq-system} at the trivial equilibrium $(u_1,u_2,\ldots,u_n)=(0,0,\ldots,0)$, i.e.
\begin{equation}\label{eq-r}
\rho :=s(L+R),
\end{equation} 
where $R=\mathrm{diag}\{r_i\}$ and $L$ is the connection matrix \eqref{eq-L}. Let $\mathbf{u}^*=(u_1^*, u_2^*, \ldots, u_n^*)$ be the unique positive equilibrium of \eqref{eq-system} (see \cite{cosner1996variability,li2010global,lu1993global} for the existence and uniqueness result of the positive equilibrium of \eqref{eq-system}). The \textit{network biomass} is defined as 
\[\MM=\sum_{i=1}^n u_i^*.
\]

We assume that the movement of individuals among patches are subject to diffusion and drift, where $d > 0$ represents the diffusion magnitude and $q > 0$ represents the drift magnitude. In addition, we impose the following assumptions:
\begin{enumerate}
    \item[(A1)] $r_i\geq 0$ for all $i$, i.e. there are no patches with negative intrinsic growth rates. Moreover, $\sum_{i} r_i\neq 0$. 
    \item[(A2)] $K_i=K>0$ for all $i$, i.e. all patches have the same carrying capacity.
    \item[(A3)] The connection matrix $L$ is irreducible, and all upstream movement coefficients are
    given by $d$ and downstream movement coefficients are $d + q$.
\end{enumerate}

\subsection{Stream networks with three patches}
The main objective of this paper is to investigate the impact of \textit{edge modifications} (adding edges) on the metapopulation growth rate and network biomass. We consider two types of edge modifications: adding a single edge (one-directional edge modification) and adding a pair of edges with equal movement coefficients (bi-directional edge modification).   For most of the analysis in the paper, we focus on stream networks with three patches. While there are three possible configurations of 3-patch stream networks (straight, tributary, distributary - see \cite{jiang2021three}), edge modifications make the most sense in the cases of tributary and distributary stream networks (Figures \ref{fig-tributary} and \ref{fig-distributary}). We note that previous studies have examined the impact of one-directional edge modifications on the eigenvalues of a matrix in other applications \cite{ hadeler2008monotone, kirkland2021impact}.

\begin{figure}[htbp]
\centering
\begin{tikzpicture}[scale=0.75, transform shape]
\begin{scope}[every node/.style={draw}, node distance= 1.5 cm]

    \node[draw=white] (i) at (-5.5,-3) {(0)};
    \node[draw=white] (ii) at (0,-3) {(1)};
    \node[draw=white] (iii) at (5.5,-3) {(2)};
    \node[circle] (1) at (-7,0) {$1$};
    \node[circle] (2) at (-4,0) {$2$};
    \node[circle] (3) at (-5.5,-2) {$3$};
    \node[circle] (4) at (-1.5,0) {$1$};
    \node[circle] (5) at (1.5,0) {$2$};
    \node[circle] (6) at (0,-2) {$3$};
    \node[circle] (7) at (4,0) {$1$};
    \node[circle] (8) at (7,0) {$2$};
    \node[circle] (9) at (5.5,-2) {$3$};
\end{scope}
\begin{scope}[every node/.style={fill=white},
              every edge/.style={thick}]

    \draw[line width=0.5mm] [->](1) to [bend right] node[left=5] {{\footnotesize $d+q$}} (3);
    \draw[thick] [->](3) to node[right=4] {{\footnotesize $d$}} (1);  
    \draw[line width=0.5mm] [->](2) to [bend left] node[right=5] {{\footnotesize $d+q$}} (3);
    \draw[thick] [->](3) to node[left=4] {{\footnotesize $d$}} (2);
    \draw[line width=0.5mm] [->](4) to [bend right] node[left=5] {{\footnotesize $d+q$}} (6);
    \draw[thick] [->](6) to node[right=4] {{\footnotesize $d$}} (4);
    \draw[line width=0.5mm] [->](5) to [bend left] node[right=5] {{\footnotesize $d+q$}} (6);
    \draw[thick] [->](6) to  node[left=4] {{\footnotesize $d$}} (5);
    \draw[thick,dashed] [->] (4) to  node[above=2] {{\footnotesize $k$}} (5);
    \draw[line width=0.5mm] [->](7) to [bend right] node[left=5] {{\footnotesize $d+q$}} (9);
    \draw[thick] [->](9) to node[right=4] {{\footnotesize $d$}} (7);
    \draw[line width=0.5mm] [->](8) to [bend left] node[right=5] {{\footnotesize $d+q$}} (9);
    \draw[thick] [->](9) to  node[left=4] {{\footnotesize $d$}} (8);
    \draw[thick,dashed] [->] (8) to [bend right] node[above=5] {{\footnotesize $k$}} (7);
    \draw[thick,dashed] [->] (7) to  node[below=2] {{\footnotesize $k$}} (8);
\end{scope}
\end{tikzpicture}
\caption{(0) The tributary stream network, where 1,2 are upstream nodes and 3 is a downstream node; (1) the tributary stream network with one-directional edge modification; (2) the tributary stream network with bi-directional edge modification.
}\label{fig-tributary}
\end{figure}
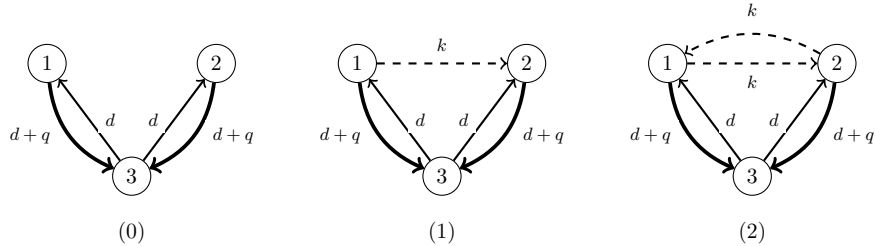

Figure \ref{fig-tributary} illustrates the tributary stream network with two types of edge modification. 
The 3-patch tributary stream has connection matrix
    \[
    L_\mathcal{T}=\begin{bmatrix}
    -d-q & 0 & d \\
    0 & -d-q & d \\
    d+q & d+q & -2d \\
    \end{bmatrix}.
    \]
The modified connection matrices are $L_{\mathcal{T},m}(k) = L_\mathcal{T} + kE_{\TT,m}$, where $m=1,2$ correspond to the one-directional and bi-directional edge modifications. In particular, the matrices $E_\TT^m$ are given as follows:
\[
    E_{\TT,1} =\begin{bmatrix}
    -1 & 0 & 0 \\
    1 & 0 & 0 \\
    0 & 0 & 0 \\
    \end{bmatrix}, \quad  E_{\TT,2} =\begin{bmatrix}
    -1 & 1 & 0 \\
    1 & -1 & 0 \\
    0 & 0 & 0 \\
    \end{bmatrix}.
\]
We denote the metapopulation growth rate and network biomass of the modified tributary stream network by $\rho_{\TT,m}(k)$ and $\MM_{\TT,m}(k)$ for $m=1, 2$.

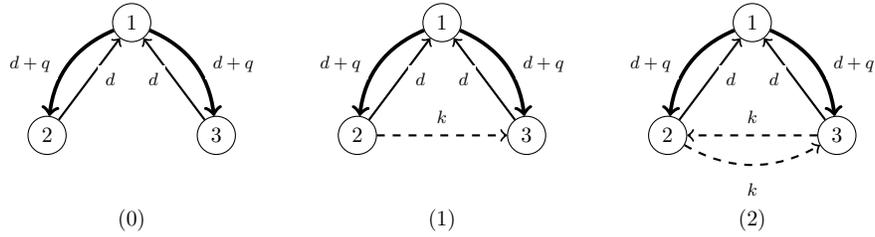
\begin{figure}[htbp]
\centering
\begin{tikzpicture}[scale=0.75, transform shape]
\begin{scope}[every node/.style={draw}, node distance= 1.5 cm]

    \node[draw=white] (i) at (-5.5,-3.5) {(0)};
    \node[draw=white] (ii) at (0,-3.5) {(1)};
    \node[draw=white] (iii) at (5.5,-3.5) {(2)};
    \node[circle] (1) at (-5.5,0) {$1$};
    \node[circle] (2) at (-7,-2) {$2$};
    \node[circle] (3) at (-4,-2) {$3$};  
    \node[circle] (4) at (0,0) {$1$};
    \node[circle] (5) at (-1.5,-2) {$2$};
    \node[circle] (6) at (1.5,-2) {$3$};
    \node[circle] (7) at (5.5, 0) {$1$};
    \node[circle] (8) at (4,-2) {$2$};
    \node[circle] (9) at (7,-2) {$3$};
\end{scope}
\begin{scope}[every node/.style={fill=white},
              every edge/.style={thick}]

    \draw[line width=0.5mm] [->](1) to [bend right] node[left=5] {{\footnotesize $d+q$}} (2);
    \draw[thick] [->](2) to node[right=4] {{\footnotesize $d$}} (1);  
    \draw[line width=0.5mm] [->](1) to [bend left] node[right=5] {{\footnotesize $d+q$}} (3);
    \draw[thick] [->](3) to node[left=4] {{\footnotesize $d$}} (1);
    \draw[line width=0.5mm] [->](4) to [bend right] node[left=5] {{\footnotesize $d+q$}} (5);
    \draw[thick] [->](5) to node[right=4] {{\footnotesize $d$}} (4);
    \draw[line width=0.5mm] [->](4) to [bend left] node[right=5] {{\footnotesize $d+q$}} (6);
    \draw[thick] [->](6) to  node[left=4] {{\footnotesize $d$}} (4);
    \draw[thick,dashed] [->] (5) to  node[above=2] {{\footnotesize $k$}} (6);
    \draw[line width=0.5mm] [->](7) to [bend right] node[left=5] {{\footnotesize $d+q$}} (8);
    \draw[thick] [->](8) to node[right=4] {{\footnotesize $d$}} (7);
    \draw[line width=0.5mm] [->](7) to [bend left] node[right=5] {{\footnotesize $d+q$}} (9);
    \draw[thick] [->](9) to  node[left=4] {{\footnotesize $d$}} (7);
    \draw[thick,dashed] [->] (8) to [bend right] node[below=5] {{\footnotesize $k$}} (9);
    \draw[thick,dashed] [->] (9) to  node[above=2] {{\footnotesize $k$}} (8);
\end{scope}
\end{tikzpicture}
\caption{(0) The distributary stream network, where 1 is an upstream node and 2, 3 are downstream nodes; (1) the distributary stream network with one-directional edge modification; (2) the distributary stream network with bi-directional edge modification.
}\label{fig-distributary}
\end{figure}

Figure \ref{fig-distributary} illustrates the distributary stream network with two types of edge modification. 
The 3-patch distributary stream has connection matrix
    \[
    L_\mathcal{D} =\begin{bmatrix}
  -2d-2q & d & d \\
  d+q & -d & 0 \\
  d+q & 0 & -d \\ 
\end{bmatrix}.
    \]
The modified connection matrices are $L_{\DD,m}(k) = L_\DD + kE_{\DD,m}$, where $m=1,2$ correspond to the one-directional and bi-directional edge modifications. In particular, the matrices $E_{\DD,m}$ are given as follows:
\[
    E_{\DD,1} =\begin{bmatrix}
    0 & 0 & 0 \\
    0 & -1 & 0 \\
    0 & 1 & 0 \\
    \end{bmatrix}, \quad  E_{\DD,2} =\begin{bmatrix}
    0 & 0 & 0 \\
    0 & -1 & 1 \\
    0 & 1 & -1\\
    \end{bmatrix}.
\]
We denote the metapopulation growth rate and network biomass of the modified distributary stream network by $\rho_{\DD,m}(k)$ and $\MM_{\DD,m}(k)$ for $m=1, 2$.

\section{Impact of edge modifications on the metapopulation growth rate}\label{sec:growth}

In this section, we investigate the effects of the two types of edge modification on the metapopulation growth rate. We first present the results on the bi-directional edge modification.

\subsection{Bi-directional edge modification}

First, we consider a more general setting. Let $L$ be an $n\times n$ irreducible connection matrix. Consider the modified connection matrix
\[
L(k) = L + kE,
\]
where $E$ is also an $n\times n$  connection matrix. Then it is easy to check that $L(k)$ is still an $n\times n$ irreducible connection matrix. Let $\rho(k)$ be the metapopulation growth rate corresponding to the modified connection matrix $L(k)$. By the Perron-Frobenius Theorem \cite{berman1994nonnegative}, $\rho(k)$ is an eigenvalue of $L(k)+R$ and it satisfies the following eigenvalue problem:
\begin{equation}\label{eq:rhok}
    (L(k)+R)\bm x = \rho(k) \bm x,
\end{equation}
where $\bm x$ can be chosen to be a normalized positive vector. 
Next, we state a technical lemma which allows us to compute $\rho'(k)$.

\begin{lemma}\label{lemma:rhok'}
Assume that there exists a vector $\bm v\in\mathbb{R}^n$ satisfying $v_i(\ell_{ij}+ke_{ij})=v_j(\ell_{ji}+ke_{ji})$ for all $i, j$ (or equivalently $\text{diag}({\bm v})L(k)$ is symmetric). Then we have
\[
\rho'(k)=\frac{(E\bm x\circ \bm x)\cdot \bm v}{(\bm x\circ \bm x)\cdot \bm v},
\] 
where $\circ$ represents the Hadamard product of two vectors.
\end{lemma}
\begin{proof}

On one hand, we take the Hadamard product of both sides of equation \eqref{eq:rhok} with $\bm x'$ to obtain
\begin{equation}\label{eq:xx'}
(L+kE+R)\bm x\circ\bm x' = \rho(k) (\bm x\circ\bm x').
\end{equation}

On the other hand, Differentiating equation \eqref{eq:rhok} in terms of $k$ yields
\begin{equation}\label{eq:rhok'}
    (L+kE+R)\bm x'+E \bm x = \rho'(k) \bm x + \rho(k) \bm x'.
\end{equation}
We take the Hadamard product of both sides of equation \eqref{eq:rhok'} with $\bm x$ to obtain
\begin{equation}\label{eq:x'x}
(L+kE+R)\bm x'\circ\bm x + E\bm x\circ \bm x = \rho'(k) (\bm x\circ\bm x) +\rho(k)(\bm x'\circ\bm x).
\end{equation}
Since $R$ is a diagonal matrix, we have $(R\bm x)\circ \bm x' = (R\bm x')\circ\bm x$. Thus subtracting equation \eqref{eq:xx'} from equation \eqref{eq:x'x} yields
\begin{equation}\label{eq:x'x-xx'}
\rho'(k)(\bm x\circ\bm x) = E\bm x\circ \bm x + (L+kE)\bm x'\circ\bm x - (L+kE)\bm x\circ\bm x'.
\end{equation}
We take the dot product of both sides of equation \eqref{eq:x'x-xx'} with $\bm v$ to get
\begin{equation}\label{eq:dotv}
\rho'(k)(\bm x\circ\bm x)\cdot \bm v = (E\bm x\circ \bm x)\cdot\bm v + (L+kE)\bm x'\circ\bm x \cdot\bm v- (L+kE)\bm x\circ\bm x'\cdot\bm v.
\end{equation}
We notice that
\begin{align*}
   (L+kE)\bm x'\circ\bm x \cdot\bm v  &= \sum_{i,j} v_i (\ell_{ij}+ke_{ij})x'_j x_i \\
   &=\sum_{i,j} v_j (\ell_{ji}+ke_{ji})x'_i x_j\\
   &=\sum_{i,j} v_i(\ell_{ij}+ke_{ij})x_j x'_i\\
   &=(L+kE)\bm x\circ\bm x'\cdot\bm v,
\end{align*}
where we obtain the second equality by switching indices and the third equality by using the assumption on $\bm v$. Thus plugging this into equation \eqref{eq:dotv} gives
\[
\rho'(k)(\bm x\circ\bm x)\cdot \bm v = (E\bm x\circ \bm x)\cdot\bm v \implies \rho'(k)=\frac{(E\bm x\circ \bm x)\cdot \bm v}{(\bm x\circ \bm x)\cdot \bm v}.
\]
\end{proof} 

We have the following results for the metapopulation growth rate of the modified 3-patch stream networks, which state that the population growth rate is decreasing with respect to the modification rate $k$. 
\begin{theorem}\label{thm:rho'-3patch}
The following statements hold:
\begin{enumerate}[(a)]
    \item For the tributary stream network, we have
    \[
    \rho'_{\TT,2}(k) \leq 0,
    \]
    where the equality happens if and only if $r_1=r_2$.
    \item For the distributary stream network, we have
    \[
    \rho'_{\DD,2}(k) \leq 0,
    \]
    where the equality happens if and only if $r_2=r_3$.
\end{enumerate}
\end{theorem}
\begin{proof}
(a) For the tributary stream network, the modified connection matrix is
\[
L_{\TT,2}(k) = \begin{bmatrix}
    -d-q-k & k & d \\
    k & -d-q-k & d \\
    d+q & d+q & -2d \\
    \end{bmatrix}.
\]
Let $\bm v= \left(1,1,\frac{d}{d+q}\right)^T$, then 
\[
\text{diag}(\bm v) L_{\TT,2}(k) = \begin{bmatrix}
    -d-q-k & k & d \\
    k & -d-q-k & d \\
    d & d & -2d^2/(d+q) \\
    \end{bmatrix}
\]
which is a symmetric matrix. Thus by Lemma \ref{lemma:rhok'}, we have
\[
  \rho'_{\TT,2}(k)=\frac{(E\bm x\circ \bm x)\cdot \bm v}{(\bm x\circ \bm x)\cdot \bm v} = \frac{(-x_1+x_2)x_1v_1+ (x_1-x_2)x_2v_2}{x_1^2v_1+x_2^2v_2} = \frac{-(x_1-x_2)^2}{x_1^2+x_2^2} \leq 0.
\]
The equality happens if and only if $x_1=x_2$. Recall that $\bm x$ is the eigenvector satisfying the eigenvalue problem:
    \[    (L(k)+R)\bm x = \rho(k) \bm x.\]
So $x_1, x_2$ satisfy
\begin{align*}
    &(-d-q-k+r_1)x_1+kx_2+dx_3=\rho(k)x_1\\
    &kx_1+(-d-q-k+r_2)x_2+dx_3=\rho(k)x_2.
\end{align*}
If $x_1=x_2$, then subtracting one equation from the other yields $r_1=r_2$. On the other hand, if $r_1=r_2$, patch 1 and 2 are identical and by a reordering argument, we must also have $x_1=x_2$. Thus $  \rho'_{\TT,2}(k)=0$ if and only if $r_1=r_2$.

(b) For the distributary stream network,  the modified connection matrix is
\[
L_{\DD,2}(k) = \begin{bmatrix}
-2d-2q & d & d \\
  d+q & -d-k & k \\
  d+q & k & -d-k \\ 
    \end{bmatrix}.
\]
Let $\bm v= \left(1,\frac{d}{d+q},\frac{d}{d+q}\right)^T$, then 
\[
\text{diag}(\bm v) L_{\DD,2}(k) = \begin{bmatrix}
-2d-2q & d & d \\
  d & (-d-k)d/(d+q) & kd/(d+q) \\
  d & kd/(d+q) & (-d-k)d/(d+q) \\ 
    \end{bmatrix}
\]
which is a symmetric matrix. Again, applying Lemma \ref{lemma:rhok'} and using a similar argument for when the equality happens, we obtain the desired result.
\end{proof}

\subsection{One-directional edge modification}

 For the one-directional edge modification, we first observe that the matrices $E_{\TT,1}$ and $E_{\DD,1}$ correspond to rank one perturbations of the stream networks $L_\mathcal{T}+R$ and $L_\mathcal{D}+R$, respectively. Therefore, these edge modifications fall under the scenario studied in \cite{hadeler2008monotone}. Applying Theorem 1 from \cite{hadeler2008monotone}, we have that the metapopulation growth rates $\rho_{\TT,1}(k)$ and $\rho_{\DD,1}(k)$  of the perturbed matrices $L_{\mathcal{T}, 1}(k) + R$  and  $L_{\mathcal{D}, 1}(k) + R$ are either strictly increasing, strictly decreasing, or constant. To determine the condition for each case, it suffices to compare the growth rates at $k=0$ and at $k\to\infty$.

Following \cite{hadeler2008monotone}, for the tributary stream network where the one-directional edge is added from patch 1 to patch 2 (see Figure \ref{fig-tributary}(1)), we consider two matrices
\[
A= L_{\TT,1}(0) + R = L_\TT+R = \begin{bmatrix}
    -d-q +r_1& 0 & d \\
    0 & -d-q +r_2& d \\
    d+q & d+q & -2d+r_3 \\
    \end{bmatrix},
\]
and
\[
B= \begin{bmatrix}
    -d-q +r_2& 2d \\
    d+q & -2d+r_3 \\
    \end{bmatrix},
\]
where $B$ is obtained from $A$ by adding row $1$ to row $2$, then deleting the first row and column of the resulting matrix. Informally, the matrix $B$ corresponds to the case when $k\to\infty$. 

\begin{proposition}[{\cite[Theorem 1]{hadeler2008monotone}}]\label{prop:hadeler}
The following statements hold:
\begin{enumerate}[(a)]
    \item If $s(A)=s(B)$, then $\rho'_{\TT,1}(k)=0$ for all $k\ge 0$;
    \item If $s(A)>s(B)$, then $\rho'_{\TT,1}(k)<0$ for all $k\ge 0$;
    \item If $s(A)<s(B)$, then $\rho'_{\TT,1}(k)>0$ for all $k\ge 0$.
\end{enumerate}
\end{proposition}
To compare $s(A)$ and $s(B)$, we first state some technical lemmas. The first lemma is a well-known result (e.g., see \cite[Corollary 2.1.5]{berman1994nonnegative}) .
\begin{lemma}\label{lem:PQ}
Suppose that $P$ and $Q$ are $n\times n$ real-valued matrices, $P$ is essentially
nonnegative (i.e., all off-diagonal entries are nonnegative), $Q$ is nonnegative and nonzero, and $P +Q$ is irreducible. Then $s(P+Q) > s(P)$.
\end{lemma}

\begin{lemma}\label{lem:SB}
    We have $s(B)>\max\{-d-q+r_2,-2d+r_3\}$.
\end{lemma}
\begin{proof}
We have
\[
B= \begin{bmatrix}
    -d-q +r_2& 2d \\
    d+q & -2d+r_3 \\
    \end{bmatrix} = \begin{bmatrix}
    -d-q +r_2& 0 \\
    d+q & -2d+r_3 \\
    \end{bmatrix} +\begin{bmatrix}
    0& 2d \\
    0 & 0 \\
    \end{bmatrix}. 
\]
Let $P$ and $Q$ be the two matrices in the above sum, then they satisfy the conditions in Lemma \ref{lem:PQ}. Thus $s(B) > s(P) = \max\{-d-q+r_2,-2d+r_3\}$.
\end{proof}

For the tributary stream network, the following result holds.
\begin{theorem}\label{thm:rho'_onedir_tributary}
The following statements hold for the tributary stream network:
\begin{enumerate}[(a)]
    \item If $r_1=r_2$, then $\rho'_{\TT,1}(k) = 0$ for all $k\ge 0$;
    \item If $r_1>r_2$, then $\rho'_{\TT,1}(k) <0$ for all $k\ge 0$;
    \item If $r_1<r_2$, then $\rho'_{\TT,1}(k) >0$ for all $k\ge 0$.
\end{enumerate}   
\end{theorem}
\begin{proof}
    We consider the characteristic polynomial of the matrix $A$
    \begin{align*}
    p_A (\lambda)&= \begin{vmatrix}
    -d-q +r_1-\lambda & 0 & d \\
    0 & -d-q +r_2 -\lambda & d \\
    d+q & d+q & -2d+r_3-\lambda  \\
    \end{vmatrix} \\
    &= \begin{vmatrix}
    -d-q +r_1-\lambda & 0 & d \\
    -d-q +r_1-\lambda & -d-q +r_2 -\lambda & 2d \\
    d+q & d+q & -2d+r_3-\lambda  \\
    \end{vmatrix}\\
    &=(-d-q+r_1-\lambda)p_B(\lambda) + d\begin{vmatrix}
    -d-q +r_1-\lambda & -d-q +r_2 -\lambda  \\
    d+q & d+q  \\
    \end{vmatrix}\\
    &=(-d-q+r_1-\lambda)p_B(\lambda) + d(d+q)(r_1-r_2),
    \end{align*}
where $p_B(\lambda)$ is the characteristic polynomial of the matrix $B$. Now we consider the three cases.
\begin{enumerate}[(a)]
    \item If $r_1=r_2$, then $p(A)(\lambda)=(-d-q+r_1-\lambda)p_B(\lambda)$, thus we have
    \[
    s(A)=\max\{-d-q+r_1,s(B)\}.
    \] 
    We observe that $s(B)>\max\{-d-q+r_2,-2d+r_3\}\geq -d-q+r_2=-d-q+r_1$, where the first inequality is due to Lemma \ref{lem:SB}. Thus we must have $s(A)=s(B)$. Applying Proposition \ref{prop:hadeler} gives $\rho'_{\TT,1}(k) = 0$.
    \item Suppose that $r_1>r_2$. We plug $\lambda = s(B)$ into $p_A(\lambda)$ and obtain
    \[
    p_A(s(B)) = 0 + d(d+q)(r_1-r_2) >0.
    \]
    Since $\lim_{\lambda\to\infty}p_A(\lambda)=-\infty<0$, by the Intermediate Value Theorem, there exists a root of $p_A(\lambda)$ in $(s(B),\infty)$. Thus $s(A) > s(B)$. Applying Proposition \ref{prop:hadeler} gives $\rho'_{\TT,1}(k) <0$.
    \item Suppose that $r_1<r_2$. We plug $\lambda = s(B)$ into $p_A(\lambda)$ and obtain
    \[
    p_A(s(B)) = 0 + d(d+q)(r_1-r_2) <0.
    \]
    Note that when $\lambda > s(B)$, then by Lemma \ref{lem:SB}, we have $\lambda > -d-q+r_2>-d-q+r_1$. In addition, since $p_B$ is a quadratic function which opens upward, we have $p_B(\lambda)>0$ when $\lambda > s(B)$. Combining both observations yields
    \[
    p_A(\lambda) = (-d-q+r_1-\lambda)p_B(\lambda) + d(d+q)(r_1-r_2) <0, \quad \text{for } \lambda>s(B).
    \]
    As a result, there is no root of $p_A(\lambda)$ in $[s(B),\infty)$ and consequently, we must have $s(A)<s(B)$. Applying Proposition \ref{prop:hadeler} gives $\rho'_{\TT,1}(k) > 0$.
\end{enumerate}

\end{proof}

For the tributary stream network, the following result holds. Since the proof is similar to the proof of Theorem \ref{thm:rho'_onedir_tributary}, it is provided in the Appendix. 
\begin{theorem}\label{thm:rho'_onedir_distributary}
The following statements hold for the distributary stream network:
\begin{enumerate}[(a)]
    \item If $r_2=r_3$, then $\rho'_{\DD,1}(k) = 0$ for all $k\ge 0$;
    \item If $r_2>r_3$, then $\rho'_{\DD,1}(k) <0$ for all $k\ge 0$; 
    \item If $r_2<r_3$, then $\rho'_{\DD,1}(k) >0$.
\end{enumerate}   
\end{theorem}

By Theorems \ref{thm:rho'_onedir_tributary} and \ref{thm:rho'_onedir_distributary}, for both tributary and distributary stream networks, the metapopulation growth rate decreases when there is an added edge from the better patch (i.e. higher intrinsic growth rate) to the worse patch (i.e. lower intrinsic growth rate) and vice versa. We note that an analogous result has been obtained for a two-patch discrete-time model \cite{ackleh2022interplay}.

\section{Impact of edge modifications on the network biomass}\label{sec:biomass}

In this section, we consider the effect of the two types of edge modifications on the network biomass. Again, we start with the bi-directional edge modification, and discuss the more nuanced case of the one-directional modification later.

\subsection{Bi-directional edge modification}

We first start with the tributary stream network with bi-directional edge modification (see Figure \ref{fig-tributary}(2)). Let $\bm u^*(k)=(u_1^*(k),u_2^*(k),u_3^*(k))$ be the unique positive equilibrium of the modified system with movement matrix $L_{\TT,2}(k)=L_\TT+kE_{\TT,2}$. Recall that the network biomass for the modified system is
\[
\MM_{\TT,2}(k) = u_1^*(k)+u_2^*(k)+u_3^*(k).
\]
For notational convenience, we omit the function arguments and the star and just write $(u_1,u_2, u_3)$ for the positive equilibrium. We notice that $u_1,u_2, u_3$ satisfy
\begin{subequations}
\begin{align}
&r_1u_1\left(1-\frac{u_1}{K}\right)-(d+q+k)u_1 + ku_2+du_3 = 0, \label{eq:t1}\\
&r_2u_2\left(1-\frac{u_2}{K}\right) + ku_1-(d+q+k)u_2+du_3 = 0, \label{eq:t2}\\
&r_3u_3\left(1-\frac{u_3}{K}\right) +(d+q)u_1 +(d+q)u_2 - 2du_3=0 \label{eq:t3}.
\end{align}
\label{eq:t}
\end{subequations}
Differentiating the above equations in terms of $k$ yields the second triplet of equations
\begin{subequations}
\begin{align}
&r_1u_1'\left(1-\frac{2u_1}{K}\right)-(d+q+k)u_1'- u_1 + ku_2'+u_2+du_3' = 0,\label{eq:t1'}\\
&r_2u_2'\left(1-\frac{2u_2}{K}\right) + ku_1'+u_1-(d+q+k)u_2'-u_2+du_3' = 0,\label{eq:t2'}\\
&r_3u_3'\left(1-\frac{2u_3}{K}\right) +(d+q)u_1' +(d+q)u_2' - 2du_3'=0 \label{eq:t3'}.
\end{align}
\end{subequations}
By taking $u_1'\times \eqref{eq:t1}-u_1\times \eqref{eq:t1'}$ and similar operations on the other pair of equations we obtain
\begin{subequations}
\begin{align}
&\frac{r_1u_1^2}{K}u_1' +ku_1'u_2 - ku_1u_2' + du_1'u_3 - du_1u_3' + u_1^2- u_1u_2= 0,\label{eq:t1-t1'}\\
&\frac{r_2u_2^2}{K}u_2' +ku_2'u_1 - ku_2u_1' + du_2'u_3 - du_2u_3' + u_2^2- u_1u_2 = 0,\label{eq:t2-t2'}\\
&\frac{r_3u_3^2}{K}u_3'+ (d+q)u_3'u_1 - (d+q)u_1'u_3 + (d+q)u_3'u_2 - (d+q)u_2'u_3 =0 \label{eq:t3-t3'}.
\end{align}
\label{eq:t3}
\end{subequations}
Finally, a useful equation can be obtained by taking $\eqref{eq:t1-t1'} + \eqref{eq:t2-t2'}+\frac{d}{d+q}\eqref{eq:t3-t3'}$,
\begin{equation}\label{eq:t_sum}
 \frac{r_1u_1^2}{K}u_1'+   \frac{r_2u_2^2}{K}u_2' + \frac{d}{d+q}\frac{r_3u_3^2}{K}u_3' + (u_1-u_2)^2=0.
\end{equation}
We are now ready to provide a series of technical lemmas that will lead to Theorem \ref{thm:biomass_2dir_tributary}, the main result for the tributary stream network.
\begin{lemma}\label{lem:r3=0}
    If $r_3=0$, then $\MM_{\TT,2}'(k)=0$ for all $k\ge 0$.
\end{lemma}
\begin{proof}
    It is easy to check that when $r_3=0$, $(K,K,K(d+q)/d)$ satisfies equation \eqref{eq:t} and thus it is the unique positive equilibrium. Therefore,
    \[
    \MM_{\TT,2}(k)=K\left(2+\frac{d+q}{d}\right) \implies \MM'_{\TT,2}=0.
    \]
\end{proof}

\begin{lemma}\label{lem:tributary_bounds}
    If $r_3>0$, we must have $u_1, u_2\in\left[\frac{d}{d+q}K,K\right)$ and $u_3\in \left[K,\frac{d+q}{d}K\right)$
\end{lemma}
\begin{proof}
Here we use the same technique in \cite[Theorem 5.6]{nguyen2023maximizing}. For $\bm \phi=(\phi_1, \dots, \phi_n), \bm \psi=(\psi_1, \dots, \psi_n)\in\mathbb{R}^n$, we write $\bm \phi> \bm \psi$ if  $\bm \phi\ge \bm \psi$ and $\bm \phi\neq \bm \psi$; we denote  $\bm \phi\gg \bm \psi$ if $\phi_i>\psi_i$ for all $1\le i\le n$. Since $L_{\TT,2}(k)$ is essentially nonnegative and irreducible, the solutions of \eqref{eq-system} with $L=L_{\TT,2}(k)$ generate a strongly monotone dynamical system \cite[Theorem 4.1.1]{smith2008monotone}, that is, if $\bm u_1(t)$ and $\bm u_2(t)$ are both solutions of \eqref{eq-system}, then $\bm u_1(0)>\bm u_2(0)$ implies $\bm u_1(t)\gg \bm u_2(t)$ for all $t>0$.     

Now we consider $\bar{\bm u}=(K,K,K(d+q)/d)$. We have
\begin{align*}
&r_1\bar{u}_1\left(1-\frac{\bar{u}_1}{K}\right)-(d+q+k)\bar{u}_1 + k\bar{u}_2+d\bar{u}_3 = 0,\\
&r_2\bar{u}_2\left(1-\frac{\bar{u}_2}{K}\right) + k\bar{u}_1-(d+q+k)\bar{u}_2+d\bar{u}_3 = 0, \\
&r_3\bar{u}_3\left(1-\frac{\bar{u}_3}{K}\right) +(d+q)\bar{u}_1 +(d+q)\bar{u}_2 - 2d\bar{u}_3=r_3K\frac{d+q}{d}\left(1-\frac{d+q}{d}\right) <0.
\end{align*}
Therefore, $\bm {\bar u}$ is an upper solution of \eqref{eq-system}, and the solution $\bm u(t)$ of \eqref{eq-system} (with $L=L_{\TT,2}(k)$) with initial condition $\bm u(0)=\bm {\bar u}$ is decreasing and converges to the unique positive equilibrium $\bm u^*$ as $t\to\infty$. Hence,  we have that $\bm u^* < \bm{\bar u}$, which implies $u_1<K$, $u_2<K$, and $u_3<K(d+q)/d$.

Similarly, it is easy to check that $\underline{\bm u} = (Kd/(d+q),Kd/(d+q),K)$ is a lower solution. Thus we have $\bm u^* \geq \underline{\bm u}$, which implies $u_1\geq Kd/(d+q)$, $u_2\geq Kd/(d+q)$, and $u_3\geq K$. Unlike the previous argument which has strict inequalities, here the equality may happen if $r_1=r_2=0$.
\end{proof}

\begin{lemma}\label{lemma:r1r2}
 If $r_1>r_2$, then $u_1>u_2$ for all $k\ge0$ and vice versa.   
\end{lemma}
\begin{proof}
First, we fix the values of $r_3$ and $k$. Now $u_1-u_2$ depends continuously on $r_1$ and $r_2$. We claim that $u_1-u_2=0$ if and only if $r_1=r_2$. Indeed, if $r_1=r_2$, then node 1 and 2 are identical. By a relabeling argument and the uniqueness of positive solution of \eqref{eq:t}, we have $u_1=u_2$. On the other hand, if $u_1=u_2$, subtracting equation \eqref{eq:t2} from \eqref{eq:t1} yields 
\[
(r_1-r_2)u_1\left(1-\frac{u_1}{K}\right) =0.
\]
Since $u_1<K$ from Lemma \ref{lem:tributary_bounds}, we must have $r_1=r_2$.

Since $u_1-u_2=0$ if and only if $r_1=r_2$, $u_1-u_2$ must have the same sign in each region: $R_1=\{(r_1,r_2):0\leq r_1<r_2\}$ and $R_2=\{(r_1,r_2):0\leq r_2<r_1\}$. As a result, to determine the sign of $u_1-u_2$ in $R_1$, it suffices to consider the special case $r_1=0<r_2$. Again, subtracting equation \eqref{eq:t2} from \eqref{eq:t1} and rearranging terms yields
\[
(d+q+2k)(u_2-u_1) = r_2u_2\left(1-\frac{u_2}{K}\right) >0.
\]
Thus $u_2-u_1>0$ in region $R_1$. Similarly, for region $R_2$, we consider the special case $r_2=0<r_1$. Subtracting equation \eqref{eq:t2} from \eqref{eq:t1} and rearranging terms yields
\[
(d+q+2k)(u_2-u_1) = -r_1u_1\left(1-\frac{u_1}{K}\right) <0.
\]
Thus $u_2-u_1<0$ in region $R_2$.
\end{proof}
\begin{lemma}\label{lemma:tributary_samesign}
$u_1'+u_2'$ has the same sign as $u_3'$ for all $k\ge0$. In particular, this implies further that $\MM'_{\TT,2}$ has the same sign as $u_3'$.    
\end{lemma}
\begin{proof}
Rearranging equation \eqref{eq:t3-t3'} yields
\[
u'_3 \left(\frac{r_3u_3^2}{K}+(d+q)(u_1+u_2)\right) = (d+q)u_3(u_1'+u_2').
\]
Thus $u_3'$ has the same sign as $u_1'+u_2'$, which implies $u_3'$ has the same sign as $\MM'_{\TT,2} = u_1'+u_2'+u_3'$.
\end{proof}
\begin{lemma}\label{lemma:tributary_diagonal}
Suppose that $d>q$ and $r_3>0$. Then $u_3'=0$ for some value of $k$ if and only if $r_1=r_2$.

\end{lemma}

\begin{proof}
For the first direction, suppose there exists $k$ such that $u_3'=0$ at such value. Plug this value of $k$ into equations from \eqref{eq:t}-\eqref{eq:t3}. From equation \eqref{eq:t3'} we have
\[
(d+q)u_1'+(d+q)u_2'=0 \implies u_2'=-u_1'.
\]
If $u_1'=0$, we have $u_2'=0$ and thus from equation \eqref{eq:t1'} we have $u_1=u_2$. Subtracting equation \eqref{eq:t1} from \eqref{eq:t2} yields $r_1=r_2$. Thus it remains to prove the first direction of the lemma when $u_1'\neq 0$. Adding equations \eqref{eq:t1'} and \eqref{eq:t2'}, and using $u_2'=-u_1'$ we have
\begin{equation}\label{eq:t1'+t2'}
  u_1'\left(r_1\left(1-\frac{2u_1}{K}\right) - r_2\left(1-\frac{2u_2}{K}\right)\right)=0 \implies r_1\left(\frac{2u_1}{K}-1\right)= r_2\left(\frac{2u_2}{K}-1\right).
\end{equation}
On the other hand, subtracting equation \eqref{eq:t2} from \eqref{eq:t1} yields
\[
r_1u_1\left(1-\frac{u_1}{K}\right) - r_2u_2\left(1-\frac{u_2}{K}\right) = (d+q+2k)(u_1-u_2).
\]

Suppose $u_1>u_2$. Then we have
\begin{equation}\label{eq:t1-t2}
   r_1u_1\left(1-\frac{u_1}{K}\right)>r_2u_2\left(1-\frac{u_2}{K}\right). 
\end{equation}
Note that from Lemma \ref{lem:tributary_bounds} we have $u_1,u_2\in \left[\frac{d}{d+q}K,K\right)$. Since $d>q$, we have $u_1,u_2\in (K/2,K)$ and thus both sides of equation \eqref{eq:t1'+t2'} are positive. Thus from equations \eqref{eq:t1'+t2'} and \eqref{eq:t1-t2} we have
\[
\frac{r_1u_1\left(1-\frac{u_1}{K}\right)}{r_1\left(\frac{2u_1}{K}-1\right)} > \frac{r_2u_2\left(1-\frac{u_2}{K}\right)}{r_2\left(\frac{2u_2}{K}-1\right)} \implies \frac{u_1\left(1-\frac{u_1}{K}\right)}{\frac{2u_1}{K}-1} > \frac{u_2\left(1-\frac{u_2}{K}\right)}{\frac{2u_2}{K}-1}. 
\]
It is straightforward to check that the function $f(x)=\frac{x\left(1-\frac{x}{K}\right)}{\frac{2x}{K}-1}$ is decreasing for $x\in(K/2,K)$. Thus we must have $u_1<u_2$ which is a contradiction.
Assuming $u_1<u_2$ and using the same argument lead to a contradiction as well. Thus we must have  $u_1=u_2$, which again implies $r_1=r_2$.

For the second direction, suppose that $r_1=r_2$. Now node 1 and 2 are identical, so by a relabeling argument we have $u_1=u_2$. Thus from equation \eqref{eq:t_sum} we have 
\[
\frac{r_1u_1^2}{K}(u_1'+u_2') + \frac{d}{d+q}\frac{r_3u_3^2}{K}u_3'=0.
\]
Thus $u_3'$ and $u_1'+u_2'$ must have opposite sign. Combining this with Lemma \ref{lemma:tributary_samesign} we must have  $u_3'=0$.
\end{proof}

\begin{lemma}\label{lem:tributary_u3'}
Suppose that $d>q$ and $r_3>0$. Then $u_3'>0$ for all $k\ge0$ if one of the following conditions holds:
\begin{enumerate}
    \item[(i)] $r_1=0$ and $r_2>0$;
    \item[(ii)] $r_2=0$ and $r_1>0$.
\end{enumerate}
\end{lemma}
\begin{proof}
    Suppose that $r_1=0$, $r_2>0$ and $r_3>0$. From Lemma \ref{lemma:tributary_diagonal}, we must have $u_3' >0$ or $u_3'<0$ for all $k\ge0$. Assume by contradiction that $u_3'<0$. 
    Adding equations \eqref{eq:t1'}-\eqref{eq:t3'},
    we have
    \begin{equation}\label{eq:t1+t2+t3}
          r_2u_2'\left(1-\frac{2u_2}{K}\right)+r_3u_3'\left(1-\frac{2u_3}{K}\right)=0.  
    \end{equation}
    From Lemma \ref{lem:tributary_bounds}, we have $u_2\geq \frac{d}{d+q}K>\frac{K}{2}$, and $u_3\geq K$. Thus from equation \eqref{eq:t1+t2+t3}, we have $u_2'>0$. Now rewriting equations \eqref{eq:t2'}, we have
    \[
    ku_1' = u_2'\left(r_2\left(\frac{2u_2}{K}-1\right)+(d+q+k)\right)-du_3'+(u_2-u_1).
    \]
    Since $u_2>\frac{K}{2}$, we have $\left(d+q+k+r_2\left(\frac{2u_2}{K}-1\right)\right)>0$. Additionally, since $r_1<r_2$, from Lemma \ref{lemma:r1r2} we have $u_1<u_2$. Thus
    \[
    ku_1' =\left(d+q+k+r_2\left(\frac{2u_2}{K}-1\right)\right)u_2'-du_3' +(u_2-u_1) >0 \implies u_1'>0.
    \]
    Now we have  $u_1'+u_2' >0$ and $u_3'<0$, which contradicts Lemma \ref{lemma:tributary_samesign}. Thus we must have $u_3'>0$. Due to symmetry, the same argument can be applied to the case when $r_2=0$ and $r_1>0$.

\end{proof}
\begin{theorem}\label{thm:biomass_2dir_tributary}
For the tributary stream network, if $d > q$ (diffusion is faster than advection), then we have
\[
\MM'_{\TT,2}(k)\geq 0\ \ \text{for all}\ k\ge 0,
\]
where the equality happens if and only if $r_1=r_2$ or $r_3=0$.
\end{theorem}

\begin{proof}
From Lemma \ref{lem:r3=0}, if $r_3=0$ we have   $\MM'_{\TT,2}= 0$.
Suppose that $r_3>0$. We fix the values of $r_3$ and $k$. Now $u_3'$ depends continuously on $r_1, r_2$. From Lemma \ref{lemma:tributary_diagonal}, $u_3'=0$ if and only if $r_1=r_2$. Thus $u_3'$ has the same sign in each region: $R_1=\{(r_1,r_2): 0\leq r_1<r_2\}$ and $R_2=\{(r_1,r_2): 0\leq r_2<r_1\}$. Lemma \ref{lem:tributary_u3'} already establishes that $u_3'>0$ on $\{(r_1, r_2):\ r_1=0, r_2>0\}$ and $\{(r_1, r_2): \ r_2=0, r_1>0\}$. Thus $u_3'>0$ in both regions $R_1$ and $R_2$ and $u_3'=0$ on the line $r_1=r_2$. Using Lemma \ref{lemma:tributary_samesign}, we conclude that $\MM'_{\TT,2}>0$ when $r_1\neq r_2$ and $\MM'_{\TT,2}= 0$ when $r_1=r_2$, which completes the proof. 

 \end{proof}

\begin{remark}
Through numerical simulations we observe that without the assumption $d>q$, the total network biomass can be an increasing, decreasing, or non-monotone function of $k$ (Figure \ref{fig: biomass d>q}).
\end{remark}
\begin{figure}
    \centering
    \includegraphics[width=1\linewidth]{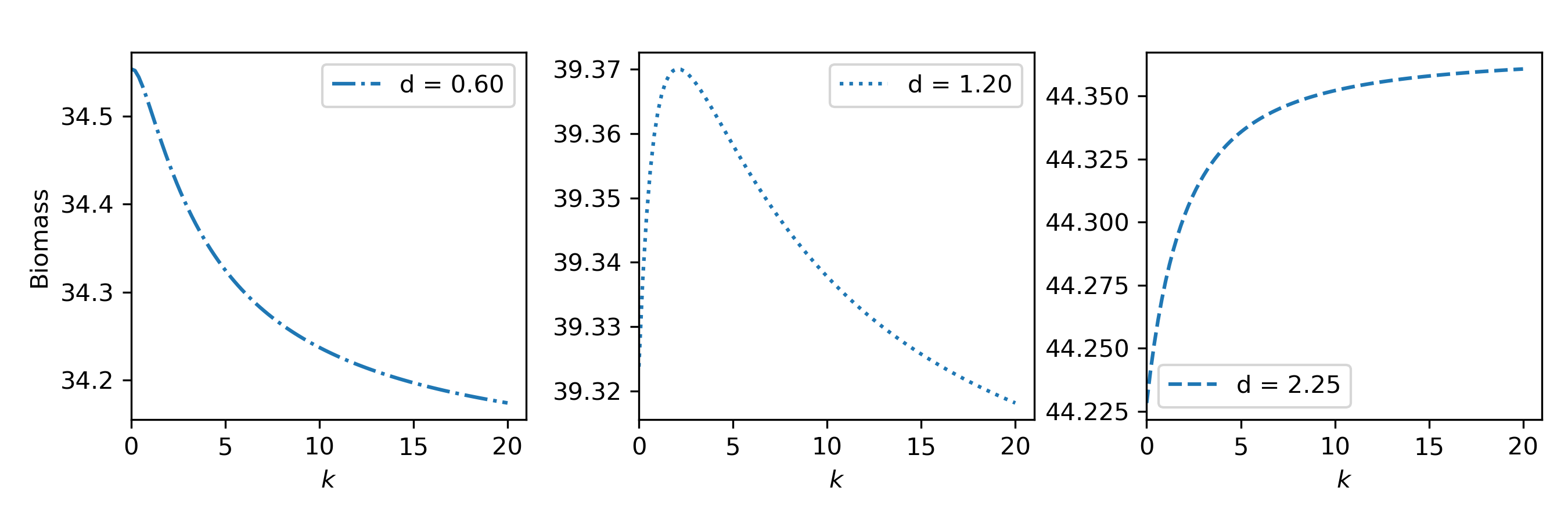}
    \caption{For the tributary network with a bi-directional edge modification, if the condition $d>q$ fails, then the network biomass can be an increasing, decreasing, or  non-monotone function of $k$. Graphs were obtained using parameters values $r_1=2.7,r_2=5.8,r_3=9.5,q=5,K=20$.}
    \label{fig: biomass d>q}
\end{figure}

The result for the distributary stream network is presented below. Notably, the monotonicity of the network biomass in this case does not depend on the relative values of $d$ and $q$. The proof for Theorem \ref{thm:biomass_2dir_distributary}, provided in the Appendix, uses similar lemmas as the proof of Theorem \ref{thm:biomass_2dir_tributary}.  

\begin{theorem}\label{thm:biomass_2dir_distributary}
For the distributary stream network,  we have
\[
\MM'_{\DD,2}(k)\leq 0,
\]
where the equality happens if and only if $r_1=0$ or $r_2=r_3$.
\end{theorem}


\subsection{One-directional edge modification}

In contrast to the case of bi-directional edge modification, we find in Theorems \ref{thm:biomass_1dir_tributary} and \ref{thm:biomass_1dir_distributary} that for $k\gtrapprox 0$ the effect of adding a one-directional edge on the total biomass depends on the values of the intrinsic growth rates $r_1, r_2, r_3$. Proofs of these theorems, which follow analogous arguments to those used in the proof of Theorems \ref{thm:biomass_2dir_tributary} and \ref{thm:biomass_2dir_distributary},  are provided in the Appendix. 
Specifically, we find that adding movement from a faster growing patch (higher intrinsic growth rate) to a slower growing patch increases the biomass. Conversely adding movement from a slower growing patch to a faster growing patch decreases the biomass.

\begin{theorem}\label{thm:biomass_1dir_tributary}
For the tributary stream network, if $d>q$, then we have
\[
    \MM'_{\TT,1}(0)  \begin{cases}
        <0 & \text{if} \quad r_1<r_2 \text{ and } r_3>0\\
        =0 & \text{if} \quad r_1=r_2 \text{ or } r_3=0\\
        >0 & \text{if} \quad r_1>r_2 \text{ and } r_3>0.\\
    \end{cases}
\]
\end{theorem}

\begin{theorem}\label{thm:biomass_1dir_distributary}
For the distributary stream network, we have
\[
    \MM'_{\DD,1}(0)  \begin{cases}
        <0 & \text{if} \quad r_2<r_3 \text{ and } r_1>0\\
        =0 & \text{if} \quad r_2=r_3 \text{ or } r_1=0\\
        >0 & \text{if} \quad r_2>r_3 \text{ and } r_1>0.\\
    \end{cases}
\]
\end{theorem}

\begin{remark}
    We note that Theorems \ref{thm:biomass_1dir_tributary}-\ref{thm:biomass_1dir_distributary} apply only for $k \gtrapprox 0$. In fact, from numerical simulation we observe that, for both network configurations, the biomass tends to decrease when $k$ is large. An example of this is provided in Figure \ref{fig: biomass 1 dir}. 
\end{remark}

\begin{figure}[h!]
    \centering
    \includegraphics[width=3in,height=2.25in]{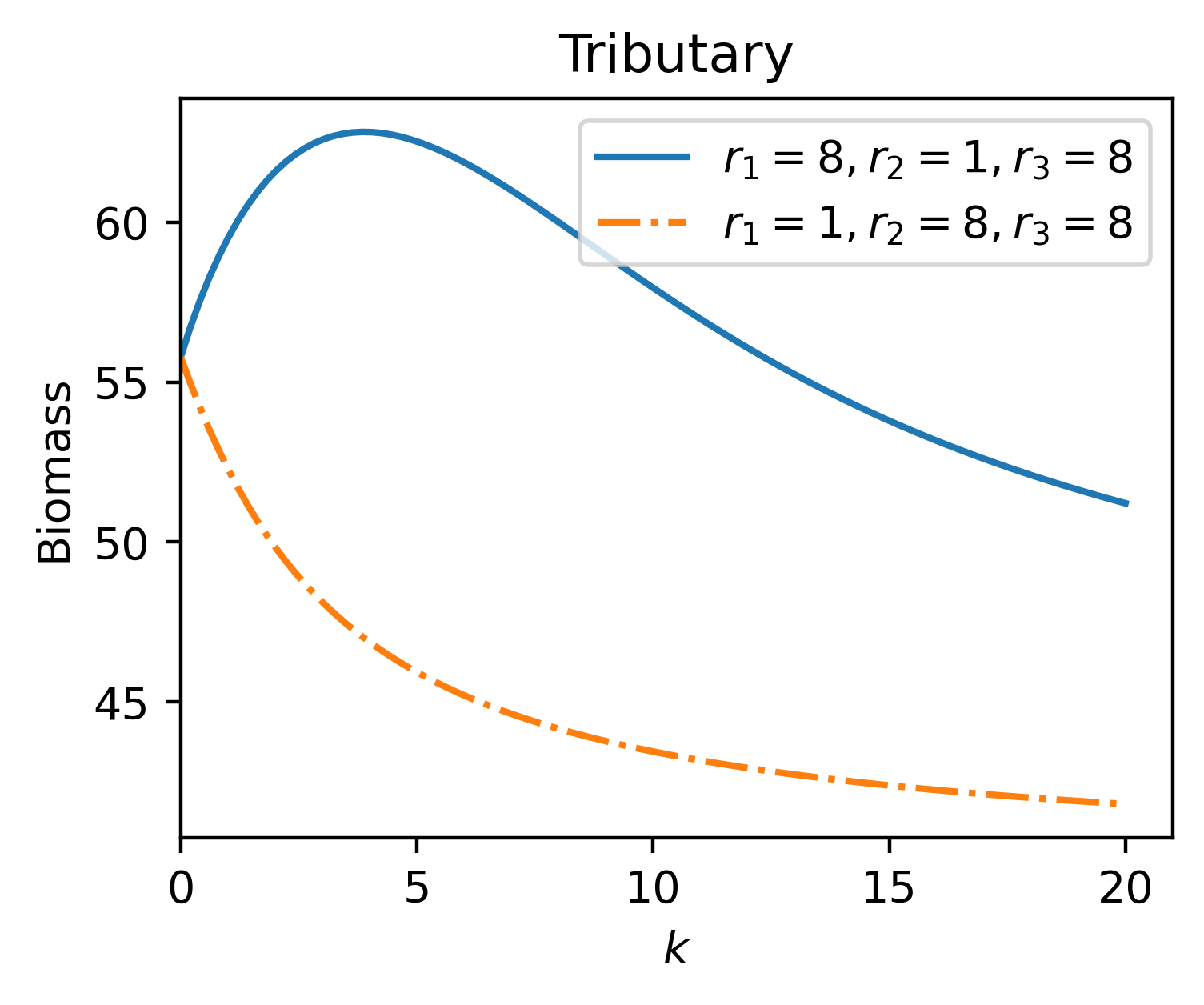}
\includegraphics[width=3in,height=2.25in]{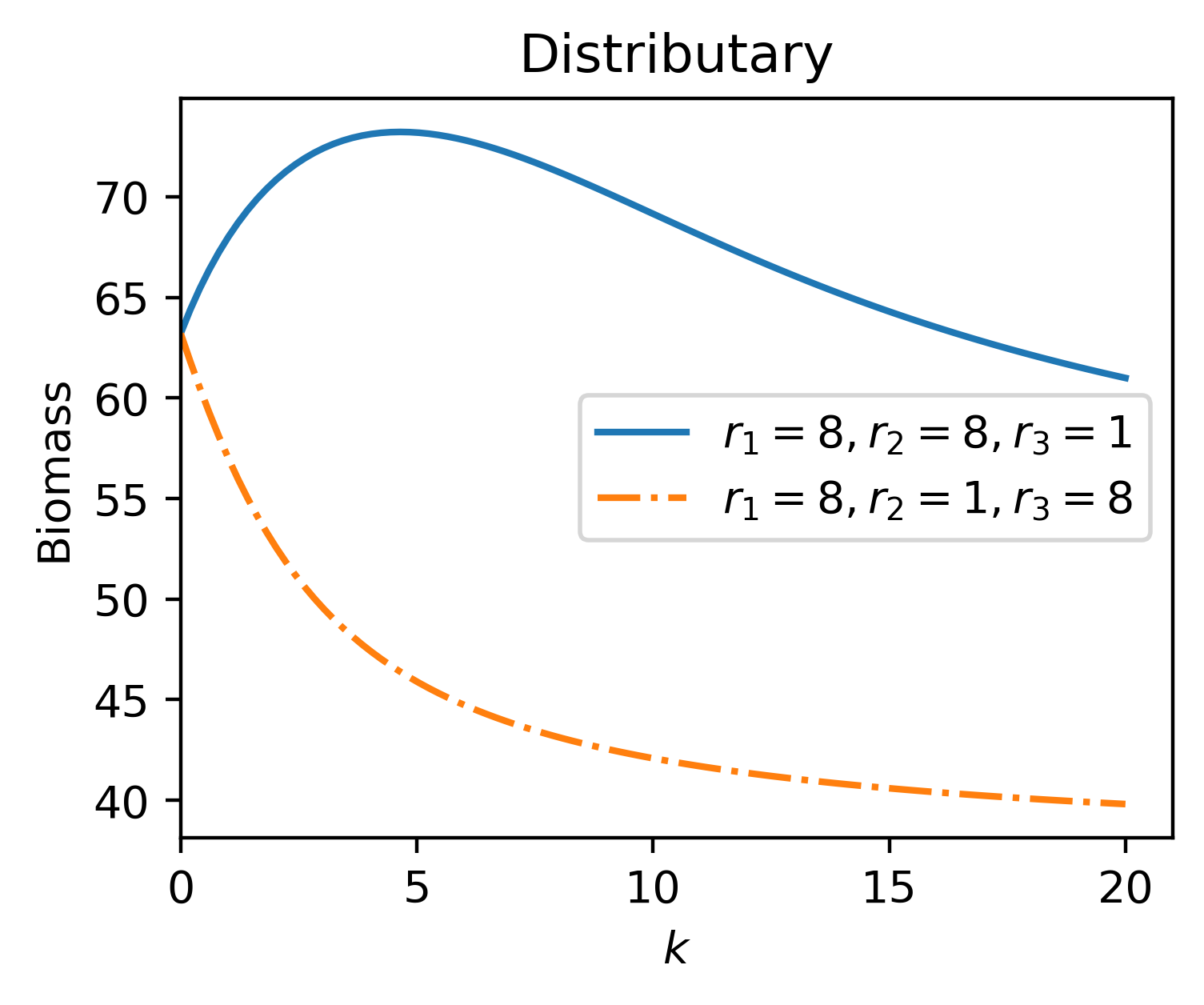}
    \caption{For the tributary and distributary networks with a one-directional edge modification, the network biomass decreases for large $k$. Graphs were obtained using parameters values  $d = 2, q=1,K=20$.}
    \label{fig: biomass 1 dir}
\end{figure}

\section{Stream networks with an arbitrary number of patches}\label{sec:homflow}
The monotonicity of the metapopulation growth rate for the case of a bi-directional edge modification (Theorem \ref{thm:rho'-3patch}) can be  extended to stream networks with any number of patches. We recall the graph-theoretic definition in \cite{nguyen2023maximizing} that allows us to describe certain stream networks with an arbitrary number of patches. 

\begin{definition}[\cite{nguyen2023maximizing}]\label{def:leveled_graph}
Let $G$ be a directed graph, and denote the set of nodes of $G$ by $V$. Consider a function $f: V\to \mathbb{Z}_{\geq 0}$. We call $f$ a \textit{level function}, $f(i)$ as the \textit{level} of node $i$ for each $i\in V$, and $(G,f)$ as a \textit{leveled graph} if the following assumptions are satisfied.
\begin{enumerate}%
    \item[(i)] For each $0\leq k\leq \max_{i\in V}\{f(i)\}$, there exists a node $j$ such that $f(j)=k$.
    \item[(ii)] For each pair of nodes $i, j$, there is no edge between $i$ and $j$ if $|f(i)-f(j)|\neq 1$.
\end{enumerate}
\end{definition}

In what follows, we associate the nodes with level 0 to be most upstream nodes, and the level of an arbitrary node with the distance between that node and the most upstream nodes. In particular, the nodes with maximum level are the most downstream nodes.  We consider homogeneous flow stream networks in which all the upstream movement coefficients are equal, and all downstream movement coefficients are equal. Two examples of homogeneous flow stream networks with five nodes are provided in Figure \ref{fig: network examples}.

\begin{definition}[\cite{nguyen2023maximizing}]
Consider a leveled graph $(G,f)$ and an irreducible connection matrix $L$. We say that $(G,f, L)$ is a \textit{homogeneous flow stream network} if the following assumptions are satisfied.
\begin{enumerate}
    \item[(i)] If there is an edge from node $i$ to node $j$, then there is also an edge from node $j$ to node $i$. 
    \item[(ii)] If there is an edge from node $i$ to node $j$, then the weight is  $\ell_{ji} = d+q$ if $f(j)-f(i)=1$ (i.e., the edge is from an upstream to a downstream node) and $\ell_{ji}=d$ if $f(i)-f(j)=1$ (i.e., the edge is from a downstream to an upstream node). 
\end{enumerate}
\end{definition}



\begin{figure}[htbp]
\centering
\begin{tikzpicture}
\begin{scope}[every node/.style={draw}, node distance= 1.5 cm]
    \node[circle] (1) at (-3+0,     0-0) {$1$};
    \node[circle] (2) at (-3-1.5,   0-2) {$2$};
    \node[circle] (3) at (-3+1.5,   0-2) {$3$};
    \node[circle] (4) at (-3-1.5,     0-4) {$4$};
    \node[circle] (5) at (-3+1.5,     0-4) {$5$};

     \node[circle] (6) at (3+0,      0-0) {$1$};
     \node[circle] (7) at (6+0,    0) {$2$};
     \node[circle] (8) at (3,    0-2) {$3$};
     \node[circle] (9) at (6,    0-2) {$4$};
     \node[circle] (10) at (3+1.5,      0-4) {$5$};
\end{scope}

\begin{scope}[every node/.style={fill=white},
              every edge/.style={thick}]
    \draw[thick] [->](1) to [bend right] node[left=5] {{}} (2); 
    \draw[thick] [<-](1) to node[right=4] {{}} (2); 
    \draw[thick]  [->](1) to [bend left] node[right=5] {{}} (3); 
    \draw[thick] [<-](1) to node[left=4] {{}} (3); 
    \draw[thick] [->](2) to [bend right]  (4); 
    \draw[thick] [<-](2) to [bend left] (4); 
    \draw[thick] [->](3) to [bend left]  (5); 
    \draw[thick] [<-](3) to [bend right] (5); 

     \draw[thick] [->](6) to [bend right]  (8); 
    \draw[thick] [->](8) to [bend right]  (6);

     \draw[thick] [->](7) to [bend right]  (9); 
    \draw[thick] [->](9) to [bend right]  (7);

       \draw[thick] [->](8) to [bend right]  (10); 
    \draw[thick] [->](10) to   (8);

       \draw[thick] [->](9) to [bend left]  (10); 
    \draw[thick] [->](10) to   (9);
\end{scope}
\end{tikzpicture}
\caption{Two leveled graphs. The left digraph has level function $f(1)=0,f(2)=f(3)=1,f(4)=f(5)=2$. The right digraph has level function $f(1)=f(2)=0,f(3)=f(4)=1,f(5)=2$. These graphs are homogeneous flow stream networks if for edges directed from an upstream node to a downstream node the weight is
$d+q$, and for edges directed from a downstream node to an upstream node the weight is $d$.} \label{fig: network examples}
\end{figure}
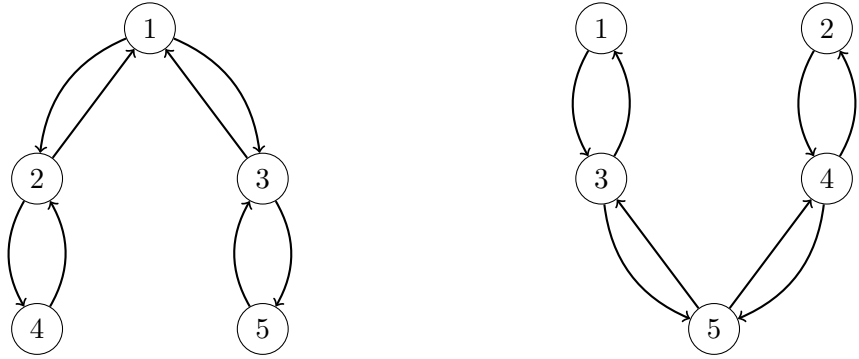

In Theorem \ref{thm:rho'-npatch}, we show that Theorem \ref{thm:rho'-3patch} can be extended to an arbitrary homogeneous flow stream network $(G,f,L)$. Consider the edge modification where we add bi-directional edges between pairs of nodes of the same level. (Notice that here we allow for possibly multiple bi-directional edges to be added as long as they are on the same level.) The modified connection matrix is
\[
L(k)=L+kE,
\]
where the matrix $E$ satisfies the following assumption.

\begin{enumerate}  \item[(A4)] \label{assumption:E-bidir} Each column of $E$ sums to zero and the off-diagonal entries ($i\neq j$) of $E$ satisfy
\begin{enumerate}[(i)]
    \item if $f(i)=f(j)$, then either {$e_{ij}=e_{ji}=1$} or $e_{ij}=e_{ji}=0$;
    \item if $f(i)\neq f(j)$ then $e_{ij}=e_{ji}=0$. 
\end{enumerate} 
\end{enumerate}
  
In Lemma \ref{lem:v}, we show that there exists a vector $\bm v$ satisfying the condition in Lemma \ref{lemma:rhok'}.

\begin{lemma}\label{lem:v}
    Consider the vector $\bm v$, where $v_i=\left(\frac{d}{d+q}\right)^{f(i)}$. Then $\bm v$ satisfies the condition in Lemma \ref{lemma:rhok'}, i.e. $v_i(\ell_{ij}+ke_{ij})=v_j(\ell_{ji}+ke_{ji})$ for all $i, j$.
\end{lemma}

\begin{proof}
From Definition \ref{def:leveled_graph}, if node $i$ and $j$ are connected in the unmodified stream network, then
\[
\frac{\ell_{ij}}{\ell_{ji}} = \begin{cases}
    \frac{d}{d+q} \quad \text{if} \quad f(i)-f(j)=1\\
    \frac{d+q}{d} \quad \text{if} \quad f(j)-f(i)=1
\end{cases}.
\]
This implies
\[
\frac{\ell_{ij}}{\ell_{ji}} = \left(\frac{d}{d+q}\right)^{f(j)-f(i)} = \frac{v_j}{v_i} \quad \text{for all } i,j,
\]
where the second equality comes from the way we define $\bm v$. Thus we have $v_i\ell_{ij}=v_j\ell_{ji}$. Since the matrix $E$ satisfies (A4), we also have 
\begin{equation*}\label{eq:E-bidir2}
  v_ie_{ij}=v_je_{ji} = \begin{cases}
    \left(\frac{d}{d+q}\right)^{f(i)}=\left(\frac{d}{d+q}\right)^{f(j)} \quad &\text{if} \quad e_{ij}=e_{ji}=1\\
    0 \quad &\text{otherwise}
\end{cases}.  
\end{equation*}
Thus we must have $v_i(\ell_{ij}+ke_{ij})=v_j(\ell_{ji}+ke_{ji})$ for all $i, j$.
\end{proof}

\begin{theorem}\label{thm:rho'-npatch}
 Let $(G,f,L)$ be a homogeneous flow stream network. Consider the edge modification $L(k)=L+kE$ where the matrix $E$ satisfies assumption (A4). Then the modified metapopulation growth rate is non-increasing, i.e. 
 \[
\rho'(k) \leq 0.
 \]
\end{theorem}
\begin{proof}
Let $S=\{(i,j)\in[1,n]\times [1,n] : e_{ij}=e_{ji}=1\}$, i.e. the set of pairs of added edges due to the matrix $E$. We can then decompose the matrix $E$ as
\[
E=\sum_{(i,j)\in S} E_{ij},
\]
where each matrix $E_{ij}$ corresponds to the modification of adding a bi-directional edge between nodes $i$ and $j$. More precisely, each matrix $E_{ij}$ contains two entries of $1$ at location $(i,j)$ and $(j,i)$, two entries of $-1$ at location $(i,i)$ and $(j,j)$, and all other entries are $0$.

Consider the vector $\bm v$ as defined in Lemma \ref{lem:v}. From Lemma \ref{lemma:rhok'}, we have
\[
\rho'(k) =  \frac{(E\bm x\circ \bm x)\cdot \bm v}{(\bm x\circ \bm x)\cdot \bm v} = \frac{\sum_{(i,j)\in S}(E_{ij}\bm x\circ \bm x)\cdot \bm v}{(\bm x\circ \bm x)\cdot \bm v}.
\]
We have
\[
(E_{ij}\bm x\circ \bm x)\cdot \bm v = (-x_i+x_j)x_iv_i+(x_i-x_j)x_jv_j =-\left(\frac{d}{d+q}\right)^{f(i)}(x_i-x_j)^2 \leq 0,
\]
where the second equality is due to the fact that each $(i,j)\in S$ must satisfy $f(i)=f(j)$ (added edges are between pairs on the same level). In addition, since all entries of $\bm v$ are positive, we must have $(\bm x\circ \bm x)\cdot \bm v =\sum_{i=1}^n v_ix_i^2>0$. Thus we can conclude
\[
\rho'(k) =\frac{\sum_{(i,j)\in S}(E_{ij}\bm x\circ \bm x)\cdot \bm v}{(\bm x\circ \bm x)\cdot \bm v}\leq 0.
\]
\end{proof}

\begin{remark}
    Using a similar argument as the one in Theorem \ref{thm:rho'-3patch}, in order for the equality in Theorem \ref{thm:rho'-npatch} to happen, for each $(i,j)\in S$, we must have node $i$ and $j$ are equivalent. More precisely, for each $(i,j)\in S$, if we relabel node $i$ and $j$, the matrix $L+R$ must be unchanged. In particular, this implies if there exists a pair $(i,j)\in S$ where $r_i\neq r_j$, then $\rho'(k)<0$.
\end{remark}

\section{Discussion and future work}\label{sec:discussion}


Overall, the findings in this paper highlight the intricate population responses to increasing connectivity in stream networks. For streams of three nodes, adding a bi-directional edge between patches always decreases the metapopulation growth rate. In contrast, its impact on biomass depends on the level of the patch at which it is added. Specifically, increased connectivity at the upstream patches have a positive impact on population biomass. This result complements the findings in \cite{nguyen2023maximizing} where it was observed that larger intrinsic growth rates in the upstream patches produce larger total biomass. Meanwhile, the addition of a bi-directional edge at a downstream patch decreases the total biomass when the diffusion rate $d$ is larger than drift rate $q$. Biologically, this corresponds to the case where the species movement rate on its own is larger than the movement rate created through the stream's current. 

With respect to one-directional edge modifications, we found that the effects on biomass and growth rate (for small $k$) are opposite. Namely, movement from a faster growing patch to a slower growing patch results in a larger total biomass (provided $d>q$ in the case of the tributary stream) but a smaller metapopulation growth rate, while the reverse is true for movement from a slower growing patch to a faster growing patch. However, we observe numerically that, in the former case, this increase in biomass only seems to hold for small $k$, with larger values of $k$ resulting in decreased total biomass irrespective of the relative growth rates (see Figure \ref{fig: biomass 1 dir}).

While the results presented in this paper primarily focus on stream networks of three nodes, we suspect that a number of these results may apply to more general homogeneous flow stream networks, as defined in Section \ref{sec:homflow}. Specifically, we make the following conjectures.
\begin{itemize}
    \item Adding string of one-directional edges to a single level decreases the metapopulation growth rate if the the intrinsic growth rates are decreasing in the direction of flow (scenarios G.1-G.6 in Table \ref{tab:my_label}). A caveat to this is that, if multiple patches are flowing into or out of this level, then this holds when these patches have the same growth rate (scenarios G.3-G.4 in Table \ref{tab:my_label}).
    \item  Adding a bi-directional edge between two sinks ($u_i^*>K)$ decreases the total biomass (scenario B.1-B.2 in Table \ref{tab:my_label}). Meanwhile, adding edges between two sources ($u_i^*<K$) increases the total biomass (scenario B.3-B.4 in Table \ref{tab:my_label}). 
    In particular, we note that nodes in the most upstream level are sources while nodes in the most downstream level are sinks \cite{nguyen2023maximizing}.
    \item Adding a one-directional edge from a faster growing (larger $r_i$) to a slower growing (smaller $r_i$) patch increases the total biomass for $k \gtrapprox 0$ (scenarios B.5-B.6 in Table \ref{tab:my_label}). Conversely, adding a one-directional edge modification from a slower growing to a faster growing patch decreases the total biomass for $k \gtrapprox 0$ (scenarios B.7-B.8 in Table \ref{tab:my_label})
    \item Some of our conjectures require certain condition on $d$ and $q$ (scenarios B.3, B.4, B.6 and B.8). From the proof of Lemma \ref{lemma:tributary_diagonal}, we conjecture that the condition is $(\frac{d}{d+q})^{n-f(i)}>\frac{1}{2}$ where $n$ is the maximum level and $i$ is one of the patches with the added edge. Thus for the scenarios considered this condition yields $(\frac{d}{d+q})^{2}>\frac{1}{2}$ which is equivalent to $d>\frac{q}{\sqrt{2}-1}$.
\end{itemize}
In Table \ref{tab:my_label}, we test these conjectures for a number of configurations consisting of four or five nodes. 
For each configuration, we apply Latin Hypercube sampling for $K = 5$, $r, d,q\in[0,10]$. Allowing $k$ to range from 0 to 10, we find that our conjectures hold for all 1,000 simulations. {We take Scenario G.1 as an example to illustrate the results. The structure `3-1' means that the configuration consists with three nodes in level 0 and one node in level 1. Nodes are always numbered sequentially  from left to right and top to bottom. So `3-1' means that the top three nodes are $\circled {1}$, $\circled {2}$, and $\circled {3}$ from left to right and the bottom node is $\circled {4}$. The new edge $``\circled {1}\rightarrow\circled {2}\rightarrow\circled {3}"$ means that we add a one-directional edge from node $\circled {1}$ to $\circled {2}$ and another one-directional edge from $\circled {2}$ to $\circled {3}$ with the same dispersal rate $k$.   For various sets of randomly assigned parameter values, the condition $r_1>r_2>r_3$ is enforced. Under this condition, we consistently observe that the metapopulation growth rate decreases as 
$k$ increases, which aligns with our conjecture. For all configurations in Table \ref{tab:my_label}, the structure defines a unique homogeneous flow stream network except for two scenarios G.3 and G.4, which are graphed in Figure \ref{fig: network examples}.  (Specifically, in the left figure of Figure \ref{fig: network examples}, this would still be a leveled graph if nodes 2 and 5 and/or nodes 3 and 4 are connected).
}

\begin{table}[!h]
    \centering
    \begin{tabular}{c|c|c|c|c}
    \toprule
        Scenario&Structure & New edge & Condition & Conjecture \\\hline 
       G.1 &$3-1$&$\circled {1}\rightarrow\circled {2}\rightarrow\circled {3}$&$r_1>r_2>r_3$&\multirow{6}{*}{Growth rate decreases}\\\cline{1-4}
       G.2 &$1-3$&$\circled {2}\rightarrow\circled {3}\rightarrow\circled {4}$&$r_2>r_3>r_4$&\\\cline{1-4}
       G.3 &$1-2-2$&$\circled {4}\rightarrow\circled {5}$&$r_4>r_5, r_2=r_3$& \\\cline{1-4}
      G.4 & $2-2-1$&$\circled {1}\rightarrow\circled {2}$&$r_1>r_2,r_3=r_4$& \\\cline{1-4}
  G.5    &  $2-1-2$&$\circled {1}\rightarrow\circled {2}$&$r_1>r_2$& \\\cline{1-4}
  G.6    &  $2-1-2$&$\circled {4}\rightarrow\circled {5}$&$r_4>r_5$& \\\hline
 B.1   &$1-1-2$&$\circled {3}\longleftrightarrow{\circled {4}}$&None&\multirow{2}{*}{Biomass monotone decreasing } \\\cline{1-4}
 
 B.2 &$2-1-2$&$\circled {4}\longleftrightarrow{\circled {5}}$&None&   \\\hline 
 
  B.3  &$2-1-1$&$\circled {1}\longleftrightarrow{\circled {2}}$&$d>\frac{q}{\sqrt{2}-1}$& \multirow{2}{*}{Biomass monotone increasing } \\\cline{1-4}
    B.4  &$2-1-2$&$\circled {1}\longleftrightarrow{\circled {2}}$&$d>\frac{q}{\sqrt{2}-1}$&  \\\hline
      
  B.5   &  $1-1-2$&$\circled {3}\rightarrow{\circled {4}}$&$r_3>r_4$& \multirow{2}{*}{Biomass initially increases} \\\cline{1-4}

       B.6    & $2-1-1$&$\circled {1}\rightarrow{\circled {2}}$&$d>\frac{q}{\sqrt{2}-1},r_1>r_2$&  \\\hline
   B.7   &  $1-1-2$&$\circled {3}\rightarrow{\circled {4}}$&$r_3<r_4$& \multirow{2}{*}{Biomass initially decreases} \\
      \cline{1-4}

          B.8    & $2-1-1$&$\circled {1}\rightarrow{\circled {2}}$&$d>\frac{q}{\sqrt{2}-1},r_1<r_2$&  \\
     \bottomrule

    \end{tabular}
    \caption{Configurations tested using  Latin Hypercube sampling with $K = 5$, $r, d,q\in[0,10]$, and 1,000 samples each. For $0\leq k\leq 10$, we find that all conjectures are supported. We list the structure as the number of patches in each level, given in increasing level order. Patches are ordered from left to right. See Figure \ref{fig: network examples} for examples of these configurations.}
    \label{tab:my_label}
\end{table}

\section{Appendix}

\subsection{Proof of Theorem \ref{thm:rho'_onedir_distributary}}
For the distributary stream network, where the bi-directional edge is added from patch 2 to patch 3 (see Figure \ref{fig-distributary}(1)), we consider two matrices
\[
A= L_{\DD,1}(0) + R = L_\DD+R = \begin{bmatrix}
  -2d-2q +r_1 & d & d \\
  d+q & -d +r_2& 0 \\
  d+q & 0 & -d +r_3\\
    \end{bmatrix},
\]
and
\[
B= \begin{bmatrix}
  -2d-2q +r_1 & d  \\
  2d+2q & -d +r_3 \\
    \end{bmatrix},
\]
where $B$ is obtained from $A$ by adding row $2$ to row $3$, then deleting the second row and column of the resulting matrix. Informally, the matrix $B$ corresponds to the case when $k\to\infty$. 

\begin{lemma}\label{lem:SB2}
    We have $s(B)>\max\{-2d-2q+r_1,-d+r_3\}$.
\end{lemma}
\begin{proof}
We have
\[
B= \begin{bmatrix}
    -2d-2q +r_1& 0 \\
    2d+2q & -d+r_3 \\
    \end{bmatrix} +\begin{bmatrix}
    0& d \\
    0 & 0 \\
    \end{bmatrix}:=P+Q. 
\]
By Lemma \ref{lem:PQ}, $s(B) > s(P) = \max\{-2d-2q+r_1,-d+r_3\}$.
\end{proof}

We are ready to prove Theorem \ref{thm:rho'_onedir_distributary}. 
\begin{proof}[Proof of Theorem \ref{thm:rho'_onedir_distributary}]
    We consider the characteristic polynomial of the matrix $A$
    \begin{align*}
    p_A (\lambda)&= \begin{vmatrix}
    -2d-2q +r_1-\lambda & d & d \\
    d+q & -d +r_2 -\lambda & 0 \\
    d+q & 0 & -d+r_3-\lambda  \\
    \end{vmatrix} \\
    &= \begin{vmatrix}
    -2d-2q +r_1-\lambda & d & d \\
    d+q & -d+r_2 -\lambda & 0 \\
    2(d+q) & -d+r_2-\lambda & -d+r_3-\lambda  \\
    \end{vmatrix}\\
    &=(-d+r_2-\lambda)p_B(\lambda) - (d+q)\begin{vmatrix}
    d & d  \\
    -d+r_2-\lambda & -d+r_3-\lambda  \\
    \end{vmatrix}\\
    &=(-d+r_2-\lambda)p_B(\lambda) + d(d+q)(r_2-r_3),
    \end{align*}
where $p_B(\lambda)$ is the characteristic polynomial of the matrix $B$. Now we consider the three cases.
\begin{enumerate}[(a)]
    \item If $r_2=r_3$, then $p(A)(\lambda)=(-d+r_2-\lambda)p_B(\lambda)$, thus we have
    \[
    s(A)=\max\{-d+r_2,s(B)\}.
    \] 
    We observe that $s(B)>\max\{-2d-2q+r_1,-d+r_3\}\geq -d+r_3=-d+r_2$, where the first inequality is due to Lemma \ref{lem:SB2}. Thus we must have $s(A)=s(B)$. Applying Proposition \ref{prop:hadeler} gives $\rho'_{\DD,1}(k) = 0$.
    
    \item Suppose that $r_2>r_3$. We plug $\lambda = s(B)$ into $p_A(\lambda)$ and obtain
    \[
    p_A(s(B)) = 0 + d(d+q)(r_2-r_3) >0.
    \]
    Since $\lim_{\lambda\to\infty}p_A(\lambda)=-\infty<0$, by the Intermediate Value Theorem, there exists a root of $p_A(\lambda)$ in $(s(B),\infty)$. Thus $s(A) > s(B)$. Applying Proposition \ref{prop:hadeler} gives $\rho'_{\DD,1}(k) <0$.
    \item Suppose that $r_2<r_3$. We plug $\lambda = s(B)$ into $p_A(\lambda)$ and obtain
    \[
    p_A(s(B)) = 0 + d(d+q)(r_2-r_3) <0.
    \]
    Note that when $\lambda > s(B)$, then by Lemma \ref{lem:SB2}, we have $\lambda > -d+r_3>-d+r_2$. In addition, since $p_B$ is a quadratic function which opens upward, we have $p_B(\lambda)>0$ when $\lambda > s(B)$. Combining both observations yields
    \[
    p_A(\lambda) = (-d+r_2-\lambda)p_B(\lambda) + d(d+q)(r_2-r_3) <0, \quad \text{for } \lambda>s(B).
    \]
    As a result, there is no root of $p_A(\lambda)$ in $[s(B),\infty)$ and consequently, we must have $s(A)<s(B)$. Applying Proposition \ref{prop:hadeler} gives $\rho'_{\DD,1}(k) > 0$.
\end{enumerate}
\end{proof}

\subsection{Proof of Theorem \ref{thm:biomass_2dir_distributary}}
Let $(u_1, u_2, u_3)$ be the unique positive solution of 
\begin{subequations}
\begin{align}
& \displaystyle -2(d+q)u_1+du_2+du_3+r_1u_1\left(1-\frac{u_1}{K}\right)  =0,\medskip \label{uk:1}\\
& \displaystyle (d+q)u_1-(d+k)u_2+ku_3+r_2u_2\left(1-\frac{u_2}{K}\right)  =0,\medskip \label{uk:2}\\
& \displaystyle (d+q)u_1+ku_2-(d+k)u_3+r_3u_3\left(1-\frac{u_3}{K}\right) =0. \label{uk:3}
\end{align}
\label{uk}
\end{subequations}

Differentiating \eqref{uk} with respect to $k$, we obtain 
\begin{subequations}
\begin{align}
& \displaystyle -2(d+q)u_1'+du_2'+du_3'+r_1u_1'\left(1-\frac{2u_1}{K}\right)=0,\medskip \label{ukp:1}\\
& \displaystyle (d+q)u_1'-(d+k)u_2'+ku_3'-u_2+u_3+r_2u_2'\left(1-\frac{2u_2}{K}\right)=0,\medskip \label{ukp:2}\\
& \displaystyle (d+q)u_1'+ku_2'-(d+k)u_3'+u_2-u_3+r_3u_3'\left(1-\frac{2u_3}{K}\right)=0.
\label{ukp:3}
\end{align}
\label{ukp}
\end{subequations}

\begin{lemma}\label{bounds}
If $r_1>0$, and $r_2>0$ or $r_3>0$, then $u_1\in(\frac{d}{d+q}K,K)$ and $u_2,u_3\in(K,\frac{d+q}{d}K)$.
\end{lemma}
\begin{proof}
    The proof is similar to Lemma \ref{lem:tributary_bounds} by noticing that $(dK/(d+q)K, K, K)$ is a strict subsolution and $(K, (d+q)K/d, (d+q)K/d)$ is a strict supersolution of \eqref{uk}. So we omit the details here. 
\end{proof}

\begin{lemma}\label{lemma_r2r3}
   Suppose that $r_1>0$. If $r_2>r_3$, then $u_3 > u_2$, and vice versa. 
\end{lemma}
\begin{proof}
    First, suppose that $k=0$. Let $f(r, x):=(d+q)u_1-dx+rx(1-x/K)$. By Lemma \ref{bounds}, we have $u_1>dK/(d+q)$. 
    For any $r>0$, $f(r, K)=(d+q)u_1-dK>0$. Since $f(r,x)$ is a parabola opening downward, $f(r, x)=0$ has a unique positive root in $(K, \infty)$.  
    In particular, $x=u_i>K$ is the unique positive solution of $f(r_i, x)=0$ for $i=2, 3$. Differentiating $f(r, x)=0$ with respect to $r$, we obtain 
    $$
  \frac{dx}{dr}\left(-d+r\left(1-\frac{2x}{K}\right)\right)=-x\left(1-\frac{x}{K}\right).
    $$
    If $x>K$, then $dx/dr<0$. Therefore, $r_2>r_3$ implies $u_2<u_3$. 
    
Next, we consider the case $k>0$. By continuity, the fact that $r_2>r_3$ implies $u_2<u_3$ is still true when $k$ is small. Suppose to the contrary that this statement is not true for any $k>0$. Then, there exists some $k>0$ such that $r_2>r_3$ and $u_2=u_3$. Substituting $u_2=u_3$ into  \eqref{uk}, we can cancel the terms involving $k$. Then using a similar argument as the case $k=0$, $r_2>r_3$ will imply $u_2<u_3$, which is a contradiction. 
\end{proof}

\begin{lemma}\label{lemma_sign}
    $u_1'$ and $u_2'+u_3'$ have the same sign for all $k\ge 0$.
\end{lemma}
\begin{proof}
    Multiplying equation \eqref{uk:1} by $u_1'$, equation \eqref{ukp:1} by $u_1$ and taking the difference, we obtain 
    $$
\left(d(u_2+u_3)+r_1\frac{u_1^2}{K}\right) u_1'=du_1(u_2'+u_3').
    $$
    Therefore, $u_1'$ and $u_2'+u_3'$ have the same sign.
\end{proof}

\begin{lemma}\label{lemma_u1p0}
    If $r_2=r_3$, then $u_2=u_3$ and $u_1'=u_2'=u_3'=0$ for all $k\ge 0$.
\end{lemma}
\begin{proof}
   Suppose $r_2=r_3$. By a relabeling argument and the uniqueness of positive solution of \eqref{uk}, $u_2=u_3$ for all $k\ge 0$. Let $(u_1, u_2, u_3)$ be the positive solution of \eqref{uk} when $k=0$. Since $u_2=u_3$, it is easy to see that $(u_1, u_2, u_3)$ is also the positive solution of \eqref{uk} for any $k>0$. Hence, $u_1'=u_2'=u_3'=0$ for all $k\ge 0$.
\end{proof}

\begin{lemma}\label{lemma_u1p}
   Suppose that $r_1>0$. If $u_1'=0$ for some $k\ge 0$, then $r_2=r_3$.
\end{lemma}
\begin{proof}
    By equation \eqref{ukp:1}, $u_2'+u_3'=0$. If $u_2'=0$, then $u_3'=0$ and $u_2=u_3$ by equation \eqref{ukp:2}. Then taking the difference of equations \eqref{ukp:2} and \eqref{ukp:3}, we have 
    $$
r_2u_2\left(1-\frac{u_2}{K}\right)= r_3u_3\left(1-\frac{u_3}{K}\right).
    $$
If $u_2=u_3=K$, then by equation \eqref{uk:2}, $u_1=dK/(d+q)$. By equation \eqref{uk:1}, $r_1=0$, which is a contradiction. So, $u_2=u_3\neq K$, which implies $r_2=r_3=0$. 

If $u_2'\neq 0$, adding equations \eqref{ukp:2} and \eqref{ukp:3}, we have 
$$
r_2u_2'\left(1-\frac{2u_2}{K}\right)+r_3u_3'\left(1-\frac{2u_3}{K}\right)=0.
$$
If follows that 
\begin{equation}\label{r2r3e}
   r_2\left(1-\frac{2u_2}{K}\right)=r_3\left(1-\frac{2u_3}{K}\right). 
\end{equation}
Suppose to the contrary that $r_2\neq r_3$. Without loss of generality, suppose that $r_2>r_3$. By Lemma \ref{lemma_r2r3}, we have $u_2<u_3$. Taking the difference of equations \eqref{uk:2} and \eqref{uk:3}, we obtain 
$$
r_2u_2\left(1-\frac{u_2}{K} \right)-r_3u_3\left(1-\frac{u_3}{K} \right)=(d+2k)(u_2-u_3)<0. 
$$
Dividing the above inequality by \eqref{r2r3e} and by Lemma \ref{bounds}, 
$$
\frac{u_2\left(1-\frac{u_2}{K}\right)}{\left(\frac{2u_2}{K}-1\right)}<\frac{u_3\left(1-\frac{u_3}{K}\right)}{\left(\frac{2u_3}{K}-1\right)}.
$$
It is straightforward to check that the function $f(x)=x(1-x/K)/(2x/K-1)$ is decreasing on $[K, (d+q)K/d]$. Then $u_2<u_3$ implies that $f(u_2)\ge f(u_3)$, which is a contradiction. 
\end{proof}

\begin{lemma}\label{lemma_boundary}
  Suppose that $r_1>0$.   If $r_2=0$ and $r_3>0$ or $r_2>0$ and $r_3=0$, then $u_1'<0$  for all $k\ge0$. 
\end{lemma}
\begin{proof}
    Suppose that $r_2=0$ and $r_3>0$. By Lemma \ref{lemma_u1p}, we have $u_1'\neq 0$. Suppose to the contrary that $u_1'>0$ for some $k\ge 0$. Multiplying \eqref{uk} by $u_i$ and \eqref{ukp} by $u_i'$ and taking the difference, we have 
  \begin{subequations}
      \begin{align}
          &\displaystyle d(u_1'u_2-u_1u_2')+d(u_1'u_3-u_1u_3')+r_1\frac{u_1^2u_1'}{K}=0,\medskip \label{ukup:1}\\
&\displaystyle -(d+q)(u_1'u_2-u_1u_2')+k(u_2'u_3-u_2u_3')+u_2(u_2-u_3)=0,\medskip \label{ukup:2}\\
&\displaystyle k(u_2u_3'-u_2'u_3)-(d+q)(u_2'u_3-u_2u_3')+u_3(u_3-u_2)+r_3\frac{u_3^2u_3'}{K}=0. \label{ukup:3}
      \end{align}
      \label{ukup}
  \end{subequations}

Multiplying equation \eqref{ukup:1} by $(d+q)/d$ and adding this to equations of \eqref{ukup:2} and \eqref{ukup:3}, we have 
$$
r_1\frac{d+q}{d}\frac{u_1^2u_1'}{K}+(u_2-u_3)^2+r_3\frac{u_3^2u_3'}{K}=0. 
$$
Noticing $u_1'>0$, the above equation implies that $u_3'<0$. By Lemma \ref{lemma_r2r3}, we have $u_2-u_3>0$. By Lemma \ref{bounds}, $u_3>K$ and $1-2u_3/K<0$. So by equation \eqref{ukp:3}, $u_2'<0$. Hence, $u_2'+u_3'<0$, which contradicts with Lemma \ref{lemma_sign}. Therefore, $u_1'<0$ for all $k\ge0$. 
\end{proof}

\begin{theorem}
    Given $r_2\neq r_3$ and $r_1>0$, then the total biomass is a decreasing function of the diffusion rate between the patches on the same level.
\end{theorem}
\begin{proof}
    By Lemma \ref{lemma_sign}, it suffices to show that $u_1'<0$ for all $k\ge 0$. From Lemmas \ref{lemma_u1p0} and \ref{lemma_u1p}, $u_1'=0$ if and only if $r_2=r_3$. Thus $u_1'$ has the same sign in each of the two regions: $R_1=\{(r_2,r_3): 0\leq r_2<r_3\}$ and $R_2=\{(r_2,r_3): 0\leq r_3<r_2\}$. By Lemma \ref{lemma_boundary}, $u_1'<0$ on $\{(r_1, r_2):\ r_2=0, r_3>0\}$ and $\{(r_1, r_2):\ r_3=0, r_2>0\}$. Thus $u_1'<0$ in both regions $R_1$ and $R_2$. This completes the proof.
\end{proof}

\subsection{Proof of Theorem \ref{thm:biomass_1dir_tributary}}
The positive equilibrium $(u_1,u_2,u_3)$ in the tributary stream network with one-directional edge modification is the unique solution of
\begin{subequations}
\begin{align}
&r_1u_1\left(1-\frac{u_1}{K}\right)-(d+q+k)u_1 +du_3 = 0 \\
&r_2u_2\left(1-\frac{u_2}{K}\right) + ku_1-(d+q)u_2+du_3 = 0 \\
&r_3u_3\left(1-\frac{u_3}{K}\right) +(d+q)u_1 +(d+q)u_2 - 2du_3=0 .
\end{align}
\end{subequations}
Differentiating the above equations in terms of $k$ yields
\begin{subequations}
\begin{align}
&r_1u_1'\left(1-\frac{2u_1}{K}\right)-(d+q+k)u_1'- u_1 +du_3' = 0\\
&r_2u_2'\left(1-\frac{2u_2}{K}\right) + ku_1'+u_1-(d+q)u_2'+du_3' = 0\\
&r_3u_3'\left(1-\frac{2u_3}{K}\right) +(d+q)u_1' +(d+q)u_2' - 2du_3'=0 .
\end{align}
\end{subequations}
Evaluating at $k=0$, we have
\begin{subequations}
\begin{align}
&r_1u_1'(0)\left(1-\frac{2u_1(0)}{K}\right)-(d+q)u_1'(0)- u_1(0) +du_3'(0) = 0 \label{eq:u1'0}\\
&r_2u_2'(0)\left(1-\frac{2u_2(0)}{K}\right) + u_1(0)-(d+q)u_2'(0)+du_3'(0) = 0 \label{eq:u2'0}\\
&r_3u_3'(0)\left(1-\frac{2u_3(0)}{K}\right) +(d+q)u_1'(0) +(d+q)u_2'(0) - 2du_3'(0)=0 .
\end{align}
\end{subequations}
Adding the three equations above yields
\begin{align}\label{eq:sum_ui'_tributary}
r_1u_1'(0)\left(1-\frac{2u_1(0)}{K}\right)+r_2u_2'(0)\left(1-\frac{2u_2(0)}{K}\right) +r_3u_3'(0)\left(1-\frac{2u_3(0)}{K}\right)  =0 
\end{align}

It is easy to check that when $k=0$, all the following results are special cases of Lemma \ref{lem:r3=0} to Lemma \ref{lemma:tributary_diagonal}. 

\begin{lemma}
    If $r_3=0$, then $\MM'_{\TT,1}(0)=0$.
\end{lemma}

\begin{lemma}
    If $r_3>0$, we must have $u_1(0),u_2(0) \in\left[\frac{d}{d+q}K,K\right)$ and $u_3(0)\in\left[K,\frac{d+q}{d}K\right)$.
\end{lemma}

\begin{lemma}
    $u_1'(0)+u_2'(0)$ has the same sign as $u_3'(0)$. In particular, this implies further that $\MM'_{\TT,1}$ has the same sign as $u_3'(0)$.
\end{lemma}

\begin{lemma}
Suppose that $d>q$ and $r_3>0$. Then $u_3'(0)=0$ if and only if $r_1=r_2$.
\end{lemma}

Now it remains to check the sign of $u_3'(0)$ on the two axes $\{r_1=0, r_2>0\}$ and $\{r_2=0, r_1>0\}$. 

\begin{lemma}
    Suppose that $d>q$ and $r_3>0$. Then
    \begin{enumerate}[(a)]
        \item  If $r_1=0$ and $r_2>0$, then $u_3'(0)<0$.
        \item  If $r_2=0$ and $r_1>0$, then $u_3'(0)>0$.
    \end{enumerate}
\end{lemma}
\begin{proof}
    (a) Suppose that $r_1=0$,  $r_2>0$, and $r_3>0$. We assume by contradiction that $u_3'(0)>0$. From equation \eqref{eq:sum_ui'_tributary}, we have
    \[
   r_2u_2'(0)\left(1-\frac{2u_2(0)}{K}\right) +r_3u_3'(0)\left(1-\frac{2u_3(0)}{K}\right)  =0 .
    \]
    Since $u_2(0)\geq \frac{d}{d+q}K>\frac{K}{2}$ and $u_3(0)\geq K >\frac{K}{2}$, we have $u_2'(0)<0$. However, from equation \eqref{eq:u2'0} we have
    \[
    u_2'(0)\left(d+q+r_2\left(\frac{2u_2(0)}{K}-1\right)\right) = u_1(0) +du_3'(0) >0.
    \]
    Thus we now obtain $u_2'(0)>0$, which is a contradiction. Therefore, for this case we must have $u_3'(0)<0$.

    (b) Suppose that $r_2=0$, $r_1>0$, and $r_3>0$. Again, we assume by contradiction that $u_3'(0)<0$. From equation \eqref{eq:sum_ui'_tributary}, we have 
     \[
   r_1u_1'(0)\left(1-\frac{2u_1(0)}{K}\right) +r_3u_3'(0)\left(1-\frac{2u_3(0)}{K}\right)  =0 .
    \]
    Since $u_1(0)\geq \frac{d}{d+q}K>\frac{K}{2}$ and $u_3(0)\geq K >\frac{K}{2}$, we have $u_1'(0)>0$. However, from equation \eqref{eq:u1'0}, we have
      \[
    u_1'(0)\left(d+q+r_2\left(\frac{2u_2(0)}{K}-1\right)\right) = -u_1(0) + du_3'(0) <0.
    \]
    Thus now we obtain $u_1'(0)<0$, which is a contradiction. Therefore, for this case we must have $u_3'(0)>0$.
\end{proof}

Since $\MM'_{\TT,1}(0)$ has the same sign as $u_3'(0)$, the proof of Theorem \ref{thm:biomass_1dir_tributary} follows directly from the lemmas above.
\subsection{Proof of Theorem \ref{thm:biomass_1dir_distributary}}

The positive equilibrium $(u_1,u_2,u_3)$ in the distributary stream network with one-directional edge modification is the unique solution of
\begin{subequations}
\begin{align}
&r_1u_1\left(1-\frac{u_1}{K}\right)-2(d+q)u_1 +du_2+du_3 = 0 ,\label{eq:u1_distributary}\\
&r_2u_2\left(1-\frac{u_2}{K}\right)+ (d+q)u_1 -(d+k)u_2 = 0, \label{eq:u2_distributary}\\
&r_3u_3\left(1-\frac{u_3}{K}\right) +(d+q)u_1 +ku_2 - du_3=0 \label{eq:u3_distributary}.
\end{align}
\end{subequations}
Differentiating the above equations in terms of $k$ yields
\begin{subequations}
\begin{align}
&r_1u_1'\left(1-\frac{2u_1}{K}\right)-2(d+q)u_1' +du_2'+du_3'  = 0,\\
&r_2u_2'\left(1-\frac{2u_2}{K}\right) + (d+q)u_1' - (d+k)u_2' - u_2  = 0,\\
&r_3u_3'\left(1-\frac{2u_3}{K}\right) +(d+q)u_1' +ku_2'+u_2 - du_3'=0 .
\end{align}
\end{subequations}
Evaluating at $k=0$, we have
\begin{subequations}
\begin{align}
&r_1u_1'(0)\left(1-\frac{2u_1(0)}{K}\right)-2(d+q)u_1'(0)+du_2'(0) +du_3'(0) = 0, \label{eq:u1'0_distributary}\\
&r_2u_2'(0)\left(1-\frac{2u_2(0)}{K}\right) + (d+q)u_1'(0)-d u_2'(0)-u_2(0)= 0, \label{eq:u2'0_distributary}\\
&r_3u_3'(0)\left(1-\frac{2u_3(0)}{K}\right) +(d+q)u_1'(0) +u_2(0) - du_3'(0)=0 \label{eq:u3'0_distributary}.
\end{align}
\end{subequations}
Adding the three equations above yields
\begin{align}\label{eq:sum_ui'_distributary}
r_1u_1'(0)\left(1-\frac{2u_1(0)}{K}\right)+r_2u_2'(0)\left(1-\frac{2u_2(0)}{K}\right) +r_3u_3'(0)\left(1-\frac{2u_3(0)}{K}\right)  =0. 
\end{align}
Plugging in $k=0$ and taking $\eqref{eq:u1_distributary}\times u_1'(0)+\eqref{eq:u2_distributary}\times u_2'(0)+\eqref{eq:u3_distributary}\times u_3'(0)-(\eqref{eq:u1'0_distributary}\times u_1+\eqref{eq:u2'0_distributary}\times u_2+\eqref{eq:u3'0_distributary}\times u_3)$ yields
\begin{align}\label{eq:sum_distributary}
\frac{r_1u_1(0)^2}{K}u_1'(0)+\frac{r_2u_2(0)^2}{K}u_2'(0)+\frac{r_3u_3(0)^2}{K}u_3'(0)=u_2(0)(u_3(0) - u_2(0)).
\end{align}
It is easy to check that when $k=0$, all the following results are special cases of Lemma \ref{bounds} to Lemma \ref{lemma_u1p}. 

\begin{lemma}
    If $r_1=0$, then $\MM_{\DD,1}(0)=0$.
\end{lemma}

\begin{lemma}
    If $r_1>0$, we must have $u_1(0) \in\left[\frac{d}{d+q}K,K\right)$ and $u_2(0), u_3(0)\in\left[K,\frac{d+q}{d}K\right)$.
\end{lemma}

\begin{lemma}\label{lem:r2r3_distributary}
    If $r_2>r_3$, then $u_3(0)>u_2(0)$ and vice versa. 
\end{lemma}

\begin{lemma}
    $u_2'(0)+u_3'(0)$ has the same sign as $u_1'(0)$. In particular, this implies further that $\MM'_{\DD,1}$ has the same sign as $u_1'(0)$.
\end{lemma}

\begin{lemma}
Suppose that $r_1>0$. Then $u_1'(0)=0$ if and only if $r_2=r_3$.
\end{lemma}

Now it remains to check the sign of $u_1'(0)$ on the two axes $\{r_2=0, r_3>0\}$ and $\{r_3=0, r_2>0\}$. 

\begin{lemma}
    Suppose that $r_1>0$. Then
    \begin{enumerate}[(a)]
        \item if $r_2=0$ and $r_3>0$, then $u_1'(0)<0$;
        \item  if $r_3=0$ and $r_2>0$, then $u_1'(0)>0$.
    \end{enumerate}
\end{lemma}
\begin{proof}
    (a) Suppose that $r_2=0$, $r_3>0$, and $r_1>0$. From Lemma \ref{lem:r2r3_distributary}, we have $u_2(0)>u_3(0)$. Now, assume by contradiction that $u_1'(0)>0$. From equation \eqref{eq:sum_distributary}, we have
    \[
\frac{r_1u_1(0)^2}{K}u_1'(0)+\frac{r_3u_3(0)^2}{K}u_3'(0)=u_2(0)(u_3(0) - u_2(0)) <0. 
    \]
    Since $u_1'(0)>0$, we must have $u_3'(0)<0$. However, from equation \eqref{eq:u3'0_distributary} we have
    \[
    u_3'(0)\left(d+r_3\left(\frac{2u_3(0)}{K}-1\right)\right) = u_2(0) +(d+q)u_1'(0) >0.
    \]
    Thus we obtain $u_3'(0)>0$, which is a contradiction. Therefore, for this case we must have $u_1'(0)<0$.

    (b) Suppose that $r_3=0$, $r_2>0$, and $r_1>0$. By Lemma \ref{lem:r2r3_distributary}, we have $u_3(0)>u_2(0)$. Again, we assume by contradiction that $u_1'(0)<0$. From equation \eqref{eq:sum_distributary}, we have 
     \[
\frac{r_1u_1(0)^2}{K}u_1'(0)+\frac{r_2u_2(0)^2}{K}u_2'(0)=u_2(0)(u_3(0) - u_2(0)) >0.
    \]
    Since $u_1'(0)<0$, we must have $u_2'(0)>0$. However, the equation \eqref{eq:u2'0_distributary}, we have
      \[
    u_2'(0)\left(d+r_2\left(\frac{2u_2(0)}{K}-1\right)\right) = -u_2(0) + (d+q)u_1'(0) <0.
    \]
    Thus now we obtain $u_2'(0)<0$, which is a contradiction. Therefore, for this case we must have $u_1'(0)>0$.
\end{proof}

Since $\MM'_{\DD,1}(0)$ has the same sign as $u_1'(0)$, the proof of Theorem \ref{thm:biomass_1dir_distributary} follows directly from the lemmas above.

\vspace{.2in}

\noindent{\bf Acknowledgment.} This project was begun at a SQuaRE at the American Institute of Mathematics. The authors thank AIM for providing a supportive and mathematically rich environment.

\vspace{.2in}
{\large\noindent{\bf Declarations}}

\noindent{\bf Conflict of interest} The authors declare that they have no conflict of interest.\\
\noindent{\bf Data Availability Statement} Data sharing is not applicable to this article as no new data were created or analyzed in this study.
\bibliographystyle{plain}
\bibliography{ref}

\end{document}